\documentclass[
final
 , nomarks
]{dmtcs-episciences}

\usepackage[utf8]{inputenc}
\usepackage{subfigure}
\usepackage{graphicx}
\usepackage{hyperref}
\usepackage{amsmath}
\usepackage{cleveref}
\usepackage{algorithmic}
\usepackage{cite}
\usepackage{tikz}
\usepackage{appendix}
\usepackage{verbatim}

\newcommand{\cost}{\ensuremath\textrm{cost}}

\newtheorem{theorem}{Theorem}
\newtheorem{lemma}{Lemma}
\newtheorem{definition}{Definition}
\newtheorem{proposition}[theorem]{Proposition}

\newtheorem{claim}[theorem]{Claim}

\usepackage[round]{natbib}

\author{Louis Dublois \and Michael Lampis \and Vangelis Th. Paschos}
  
\title{New Algorithms for Mixed Dominating Set}
\affiliation{
  Universit{\'e} Paris-Dauphine, PSL University, CNRS, LAMSADE, Paris, France}
\received{2020-10-05}
\revised{2021-01-13, 2021-04-15}
\accepted{2021-04-16}

\begin{document}

\publicationdetails{23}{2021}{1}{10}{6824}

\maketitle
\begin{abstract}
A mixed dominating set is a set of vertices and edges that dominates all
vertices and edges of a graph. We study the complexity of exact and
parameterized algorithms for \textsc{Mixed Dominating Set}, resolving some open questions.  In
particular, we settle the problem's complexity parameterized by treewidth and
pathwidth by giving an algorithm running in time $O^*(5^{tw})$ (improving the
current best $O^*(6^{tw})$), and a lower bound showing that our algorithm
cannot be improved under the Strong Exponential Time Hypothesis (SETH), even if parameterized by pathwidth
(improving a lower bound of $O^*((2-\varepsilon)^{pw})$). Furthermore, by using a simple but so far overlooked observation on the
structure of minimal solutions, we obtain branching algorithms which improve
the best known FPT algorithm for this problem, from $O^*(4.172^k)$ to
$O^*(3.510^k)$, and the best known exact algorithm, from $O^*(2^n)$ and
exponential space, to $O^*(1.912^n)$ and polynomial space.
\end{abstract}


\section{Introduction}

Domination problems in graphs are one of the most well-studied topics in
theoretical computer science. In this paper we study a variant called \textsc{Mixed Dominating Set}: we
are given a graph $G=(V,E)$ and are asked to select $D\subseteq V$ and
$M\subseteq E$ such that $|D\cup M|$ is minimized and the set $D\cup M$
dominates $V\cup E$, where a vertex dominates itself, its neighbors, and its
incident edges and an edge dominates itself, its endpoints, and all edges with
which it shares an endpoint.

The notion of \textsc{Mixed Dominating
Set}  was first introduced in 1977 by \cite{AlaviBLN77}, and has
been studied extensively in graph theory by \cite{AlaviLWZ92}, \cite{ErdosM77}, \cite{Meir78} and \cite{PeledS94a}.
See the chapter of \cite{Haynes0093827} for a survey on the
\textsc{Mixed Dominating Set} problem.  The computational complexity of \textsc{Mixed Dominating Set}
was first studied in 1993 by \cite{majumdar1993neighborhood}, where he
showed that the problem is NP-complete. The problem remains NP-complete on
split graphs by a result of \cite{ZhaoKS11} 
and on planar bipartite graphs of maximum degree 4 by a result of \cite{Manlove99}.
\cite{majumdar1993neighborhood},  \cite{LanC13}, 
\cite{Rajaatiabs-1708-00240} and \cite{MadathilPS019} showed
that the problem is polynomial-time solvable on trees, cacti, generalized
series-parallel graphs and proper interval graphs, respectively. 

\textsc{Mixed Dominating Set} is a natural variation of domination in graphs as it can be seen as a
\emph{mix} between four standard problems: \textsc{Dominating Set}, where
vertices dominate vertices; \textsc{Edge Dominating Set}, where edges dominate
edges; \textsc{Vertex Cover}, where vertices dominate edges; and \textsc{Edge
Cover}, where edges dominate vertices. In \textsc{Mixed Dominating Set} we are asked to select vertices
\emph{and} edges in a way that dominates all vertices \emph{and} edges. 
As only
the last of these four problems is in P, it is not surprising that \textsc{Mixed Dominating Set} is
NP-hard. We are therefore motivated to study approximation, exponential-time
and parameterized algorithms for this problem, and indeed this has been the
topic of several recent papers.  On the approximation algorithms side, the
problem is well-understood: \cite{Hatami07} gave a $2$-approximation
algorithm, while more recently  \cite{DudyczLM19} showed that, under the Unique Games Conjecture (see \cite{Khot02a}), no algorithm can achieve a ratio better than $2$ for 
\textsc{Edge Dominating Set}. As we explain (Proposition \ref{lem:eds}) this hardness
result easily carries over to \textsc{Mixed Dominating Set}, thus essentially settling the problem's
approximability.  Hence, in this paper we focus on parameterized and exact
algorithms.

\textsc{Mixed Dominating Set} has recently been the focus of several works in this context. With
respect to the natural parameter (the size $k$ of the solution),  an
$O^*(7.465^k)$\footnote{$O^*$ notation suppresses polynomial factors in the input size.}
algorithm was given by \cite{JainJPS17}, more recently improved to
$O^*(4.172^k)$ by \cite{XiaoS19}.  With respect to treewidth and
pathwidth, \cite{JainJPS17} gave algorithms running in $O^*(6^{tw})$ time and
$O^*(5^{pw})$ time, improving upon the $O^*(3^{tw^2})$ time algorithm of
\cite{RajaatiHDS18}. Furthermore, Jain et al. showed that no algorithm can
solve the problem in $O^*((2-\varepsilon)^{pw})$ time under the Set Cover
Conjecture (see \cite{CyganDLMNOPSW12} for more details about the Set Cover Conjecture). These works observed that it is safe to assume that the optimal
solution has a specific structure: the selected edges form a matching whose
endpoints are disjoint from the set of selected vertices.  This observation
immediately gives an $O^*(3^n)$ algorithm for the problem, which was recently
improved to $O^*(2^n)$ by \cite{MadathilPS019} by using a
dynamic programming approach, which requires $O^*(2^n)$ space.

\subsubsection*{Our results:} The state of the art summarized above motivates
two basic questions: first, can the gap in the complexity of the problem for
treewidth and pathwidth  and  the gap between the lower and upper bound for
these parameters be closed, as explicitly asked by \cite{JainJPS17}; second,
can we solve this problem faster than the natural $O^*(2^n)$ barrier?  We
answer these questions and along the way obtain an improved FPT algorithm for
parameter $k$. Specifically we show:

(i) \textsc{Mixed Dominating Set} can be solved in $O^*(5^{tw})$ time. Somewhat surprisingly, this result
is obtained by combining observations that exist in the literature: the
equivalence of \textsc{Mixed Dominating Set} to \textsc{Distance}-$2$-\textsc{Dominating Set} by
\cite{MadathilPS019}; and the algorithm of \cite{BorradaileL16} for this problem.

(ii) \textsc{Mixed Dominating Set} cannot be solved in time $O^*((5-\varepsilon)^{pw})$, under the
SETH. This is our main result on this front, and shows that our algorithm for
treewidth and the algorithm of \cite{JainJPS17} for pathwidth are optimal.

(iii) \textsc{Mixed Dominating Set} can be solved in time $O^*(1.912^n)$ and
$O^*(3.510^k)$, in both cases using polynomial space. In order to obtain these
algorithms we refine the notion of \emph{nice mixed dominating set} which was
used in \cite{JainJPS17}. In particular, we show that there always exists an
optimal mixed dominating set such that any selected vertex has at least two
\emph{private} neighbors, that is, two neighbors which are only dominated by
this vertex.  This allows us to speed the branching on low-degree vertices.



\section{Preliminaries}

We assume familiarity with the basics of parameterized complexity (e.g.
treewidth, pathwidth, and the SETH), as given by \cite{CyganFKLMPPS15}.  Let $G
= (V,E)$ be a graph with $|V| = n$ vertices and $|E| = m$ edges. For $u \in V$,
$N(u)$ denotes the set of neighbors of $u$, $d(u) = |N(u)|$ and $N[u] = N(u)
\cup \{ u \}$.  For $U \subseteq V$ and $u \in V$, we note $N_U(u) = N(u) \cap
U$ and use $d_U(u)$ to denote $|N_U(u)|$. Furthermore, for $U\subseteq V$ we
denote $N(U)=\cup_{u\in U} N(u)$. For an edge set $E'$, we use $V(E')$ to
denote the set of endpoints of $E'$. For $V' \subseteq V$, we use $G[V']$ to
denote the subgraph of $G$ induced by $V'$. An edge cover of a graph $G=(V,E)$
is a set of edges $M\subseteq E$ such that $V(M)=V$. Recall that an edge cover
$M$ of a graph $G=(V,E)$ must have size $|M|\ge \frac{V}{2}$, since each edge
can cover at most two vertices. A \emph{mixed dominating set} of a graph
$G=(V,E)$ is a set of vertices $D\subseteq V$ and edges $M\subseteq E$ such
that (i) all vertices of $V\setminus (D\cup V(M))$ have a neighbor in $D$ (ii)
all edges of $E\setminus M$ have an endpoint in $D\cup V(M)$.

We note that the minimization problem \textsc{Mixed Dominating Set} is harder than the more well-studied \textsc{Edge Dominating Set} (EDS) problem, by a reduction that preserves most parameters from an FPT viewpoint and the size of the optimal solution. 
Hence, essentially all hardness
results for the latter problem, such as its inapproximability obtained by \cite{DudyczLM19}
or its W[1]-hardness for clique-width from \cite{FominGLS10}, carry over to \textsc{Mixed Dominating Set}.

\begin{proposition}\label{lem:eds}
There is an approximation and
parameter-preserving reduction from \textsc{Edge Dominating Set} to \textsc{Mixed Dominating Set}.
\end{proposition}

\begin{proof} Given an instance $G=(V,E)$ of \textsc{EDS} we seek a set $M$ of
$k$ edges such that all edges have an endpoint in $V(M)$. We add a new vertex
$u$ connected to all of $V$ and attach to $u$ $|V|+2$ leaves.  The new graph
has a mixed dominating set of size $k+1$ if and only if $G$ has an edge
dominating set of size $k$. 
\end{proof}

We now define a restricted notion of mixed dominating set.

\begin{definition}\label{def:nice}
A \emph{nice} mixed dominating set of a graph $G = (V,E)$ is a mixed dominating set $D \cup M$ which satisfies the following: (i) $D \cap V(M) = \emptyset$ (ii) for all $u \in D$ there exists at least two private neighbors of $u$, that is, two vertices $v_1, v_2 \in V \setminus (D \cup V(M))$ with $N(v_1) \cap D = N(v_2) \cap D = \{ u \}$.
\end{definition}


We note that a similar notion of nice mixed dominating set was used in the
algorithms of \cite{JainJPS17}, with the key difference that these algorithms
do not use the fact that every vertex of $D$ must have at least two
\emph{private} neighbors, that is, two neighbors which are dominated only by
this vertex, though these algorithms use the fact that such vertices have at
least one private neighbor.


Let us now prove that restricting ourselves to nice solutions does not change
the value of the optimal. The idea of the proof is to reuse the arguments of
\cite{MadathilPS019} to obtain an optimal solution satisfying the first
property; and then while there exists $u \in D$ with at most one private
neighbor, we replace it by an edge while maintaining a valid solution
satisfying the first property.  \footnote{In the conference version of this
paper (IPEC 2020) we used a definition of nice mixed dominating set that
included the additional property that $M$ is a matching, and erroneously
claimed that an optimal solution satisfying this definition always exists. We
are grateful to an anonymous reviewer who pointed out to us that this is not
the case. As a result, we use here a definition of nice mixed dominating set
that is slightly weaker that the one in the conference version and give a
corrected version of Lemma \ref{lemmanice}.  However, the results we obtain
remain the same.}


\begin{lemma}\label{lemmanice} For any graph $G = (V,E)$ without isolated
vertices, $G$ has a mixed dominating set $D \cup M$ of size at most $k$ if and
only if $G$ has a nice mixed dominating set $D' \cup M'$ of size at most $k$.
\end{lemma}

\begin{proof}
On direction is trivial, since any nice mixed dominating set is also by definition a mixed dominating set. For the other direction, we first recall that it was shown by \cite{MadathilPS019} that if a graph has a mixed dominating of size $k$, then it also has such a set that satisfies the first condition of Definition \ref{def:nice}. Suppose then that $D \cup M$ is such that $D \cap V(M) = \emptyset$.

We will now edit this solution so that we obtain the missing desired property, namely the fact that all vertices of $D$ have two private neighbors. Our transformations will be applicable as long as there exists a vertex $u \in D$ without two private neighbors, and will either decrease the size of the solution, or decrease the size of $D$, while maintaining a valid solution satisfying the first property of Definition \ref{def:nice}. As a result, applying these transformations at most $n$ times yields a nice mixed dominating set.

Let $I = V \setminus (D \cup V(M))$. If there exists $u \in D$ with exactly one
private neighbor, let $v\in I$ be this private neighbor. We set $D' = D
\setminus \{ u \}$ and $M' = M \cup \{ (u,v) \}$ to obtain another solution.
This solution is valid because $N(u) \setminus \{ v \}$ is dominated by $(D
\cup M) \setminus \{ u \}$, otherwise $u$ would have more that one private
neighbor.

Let us now consider $u\in D$ such that $u$ has no private neighbor. If $
N(u)\subseteq D$, then we can simply remove $u$ from the solution and obtain a
beter solution (recall that $u$ is not an isolated vertex). Otherwise, let
$v\in N(u)\setminus D$. We set $D' = D\setminus \{u\}$ and $M' = M\cup
\{(u,v)\}$ to obtain another feasible solution with fewer vertices, while still
satisfying the first property.  We repeat these modifications until we obtain
the claimed solution.
%
\end{proof}

In the remainder, when considering a nice mixed dominating set $D \cup M$ of a
graph $G = (V,E)$, we will associate with it the partition $V = D \cup P \cup
I$ where $P = V(M)$ and $I = V \setminus (D \cup P)$. We will call this a nice
mds partition. We have the following properties: (i) $M$ is an edge cover of
$G[P]$ since $P = V(M)$ and $M$ is a set of edges (ii) $I$ is an independent
set because if there were two adjacent vertices in $I$ then the edge between
them would not be dominated (iii) $D$ dominates $I$ because if there was a
vertex in $I$ not dominated by $D$ it would not be dominated at all (iv) each
$u \in D$ has two private neighbors $v_1, v_2 \in I$, that is $N(v_1) \cap D =
N(v_2) \cap D = \{ u \}$. 



We also note the following useful relation. 

\begin{lemma}\label{lemmavertexcover}
For any graph $G = (V,E)$ and any nice mds partition $V = D \cup P \cup I$ of $G$, there exists a minimal vertex cover $C$ of $G$ such that $D \subseteq C \subseteq D \cup P$. 
\end{lemma}

\begin{proof}
Since $I$ is an independent set of $G$, $D \cup P$ is a vertex cover of $G$ and hence contains some minimal vertex cover. We claim that any such minimal vertex cover $C \subseteq D \cup P$ satisfies $D \subseteq C$. Indeed, for each $u \in D$ there exists two private neighbors $v_1, v_2 \notin D \cup P$. Hence, if $u \notin C$, the edge $(u,v_1)$ is not covered, contradiction. 
\end{proof}


\section{Treewidth}

We begin with an algorithm for \textsc{Mixed Dominating Set} running in time
$O^*(5^{tw})$. We rely on three ingredients: (i) the fact that \textsc{Mixed
Dominating Set} on $G$ is equivalent to \textsc{Distance-2-Dominating Set} on
the incidence graph of $G$ by a result of \cite{MadathilPS019} (ii) the
standard fact that the incidence graph of $G$ has the same treewidth as $G$
(iii) and an $O^*(5^{tw})$ algorithm (by \cite{BorradaileL16}) for
\textsc{Distance-2-Dominating Set}.

\begin{theorem}\label{theoremtreewidth} There is an $O^*(5^{tw})$-time
algorithm for \textsc{Mixed Dominating Set} in graphs of treewidth $tw$.  \end{theorem}

\begin{proof}
We are given an instance of \textsc{Mixed Dominating Set} $G=(V,E)$. We first construct
the incidence graph of $G$, which has vertex set $V\cup E$, and has an edge
between $v\in V$ and $e\in E$ if $e$ is incident on $v$ in $G$. We denote this
graph as $I(G)$. In other words, $I(G)$ is obtained by sub-dividing every edge
of $G$ once.

We now note the standard fact that $tw(I(G))\le tw(G)$. Indeed, if $G$ is a
forest, then $I(G)$ is also a forest; while if $tw(G)\ge 2$, then we can take
any tree decomposition of $G$ and for each $e=(u,v)$ we observe that it must
contain a bag with both $u$ and $v$.  We create a bag containing $u,v,e$ and
attach it to the bag containing $u,v$.  Note that this does not increase the
width of the decomposition.  We thus obtain a decomposition of $I(G)$. 

Second, as observed by \cite{MadathilPS019}, every mixed dominating set of $G$
corresponds to a distance-$2$ dominating set of $I(G)$. Recall that a
distance-$2$ dominating set of a graph is a set of vertices $D$ such that all
vertices of $V\setminus D$ are at distance at most $2$ from $D$.

Finally, we use the algorithm of \cite{BorradaileL16} to solve
\textsc{Distance}-$2$-\textsc{Dominating Set} in time $O^*(5^{tw})$ in $I(G)$,
which gives us the optimal mixed dominating set of $G$. 
\end{proof}

The main result of this section is a lower bound matching
\Cref{theoremtreewidth}. We prove that, under SETH, for all $\varepsilon
> 0$, there is no algorithm for \textsc{Mixed Dominating Set} 
 with complexity $O^*((5-\varepsilon)^{pw})$.  The starting point of our
reduction is the problem $q$-CSP-$5$ (see \cite{Lampis20}). In this problem we are
given a \textsc{Constraint Satisfaction} (CSP) instance with $n$ variables and
$m$ constraints. The variables take values in a set of size $5$, say
$\{0,1,2,3,4\}$. 
Each constraint involves at most $q$ variables and is given as a list of acceptable assignments for these variables, where an acceptable assignment is a $q$-tuple of values from the set $\{ 0, 1, 2, 3, 4 \}$ given to each of the $q$ variables. 
The following result was
shown by \cite{Lampis20} to be a natural consequence of the SETH.

\begin{lemma}[Theorem 3.1 by \cite{Lampis20}]\label{lemmacsp} If the SETH is
true, then for all $\varepsilon>0$, there exists a $q$ such that $n$-variable
$q$-CSP-$5$ cannot be solved in time $O^*((5-\varepsilon)^n)$.  \end{lemma}

Note that in the Theorem of \cite{Lampis20} it was shown that for any alphabet size $B$,
$q$-CSP-$B$ cannot be solved in time $O^*\left((B-\varepsilon)^n\right)$ under
the SETH, but for our purposes only the case $B=5$ is relevant, for two reasons: because this
corresponds to the base of our target lower bound; and because in our construction we will represent the $B = 5$ possible values for a variable with a path of five vertices in which there exists exactly five different ways of selecting one vertex and one edge among these five vertices. 
Our plan is therefore to
produce a polynomial-time reduction which, given a $q$-CSP-5 instance with $n$ variables, produces
an equivalent \textsc{Mixed Dominating Set} instance whose pathwidth is at most $n+O(1)$.  Then, the
existence of an algorithm for the latter problem running faster than
$O^*\left((5-\varepsilon)^{pw}\right)$ would give an
$O^*\left((5-\varepsilon)^{n}\right)$ algorithm for $q$-CSP-$5$, contradicting
the SETH. 

Before giving the details of our reduction let us sketch the basic ideas, which
follow the pattern of other SETH-based lower bounds which have appeared in the
literature: see \cite{HanakaKLOS18}, \cite{JaffkeJ17}, \cite{KatsikarelisLP18}, \cite{KatsikarelisLP19} and \cite{LokshtanovMS18}. 
The
constructed graph consists of a main selection part of $n$ paths of length
$5m$, divided into $m$ sections.  Each path corresponds to a variable and each
section to a constraint. The idea is that the optimal solution will follow for
each path a basic pattern of selecting one vertex and one edge among the first
five vertices and then repeat this pattern throughout the path (see Figure
\ref{FigureMainPart}).  There are $5$ natural ways to do this, so this
can represent all assignments to the $q$-CSP-5 instance. We will then add
verification gadgets to each section, connected only to the vertices of that
section that represent variables appearing in the corresponding constraint
(thus keeping the pathwidth under control), in order to check that the selected
assignment satisfies the constraint. 

\begin{figure}[h]
\centering
\begin{tikzpicture}
\node[circle] (1) at (-3,5.5){};
\node[circle,draw=black, fill=white, inner sep=0pt,minimum size=6pt] (1_0) at (-2.5,5.5){};
\node[circle,draw=black, fill=white, inner sep=0pt,minimum size=6pt] (1_1) at (-2,5.5){};
\node[circle,draw=black, fill=white, inner sep=0pt,minimum size=6pt] (1_2) at (-1.5,5.5){};
\node[circle,draw=black, fill=white, inner sep=0pt,minimum size=6pt] (1_3) at (-1,5.5){};
\node[circle,draw=black, fill=white, inner sep=0pt,minimum size=6pt] (1_4) at (-0.5,5.5){};
\node[circle,draw=black, fill=white, inner sep=0pt,minimum size=6pt] (1_5) at (0,5.5){};
\node[circle,draw=black, fill=white, inner sep=0pt,minimum size=6pt] (1_6) at (0.5,5.5){};
\node[circle] (11) at (1,5.5){};
\draw (1) -- (1_0) -- (1_1) -- (1_2) -- (1_3) -- (1_4) -- (1_5) -- (1_6) -- (11);
\draw (-3.5,5.5) node{(a)};
\node[circle,fill=black, inner sep=0pt,minimum size=6pt] (1_a) at (-2,5.5){};
\node[circle,fill=black, inner sep=0pt,minimum size=6pt] (1_b) at (0.5,5.5){};
\draw [ultra thick] (1_3) -- (1_4);

\node[circle] (2) at (-3,4.5){};
\node[circle,draw=black, fill=white, inner sep=0pt,minimum size=6pt] (2_0) at (-2.5,4.5){};
\node[circle,draw=black, fill=white, inner sep=0pt,minimum size=6pt] (2_1) at (-2,4.5){};
\node[circle,draw=black, fill=white, inner sep=0pt,minimum size=6pt] (2_2) at (-1.5,4.5){};
\node[circle,draw=black, fill=white, inner sep=0pt,minimum size=6pt] (2_3) at (-1,4.5){};
\node[circle,draw=black, fill=white, inner sep=0pt,minimum size=6pt] (2_4) at (-0.5,4.5){};
\node[circle,draw=black, fill=white, inner sep=0pt,minimum size=6pt] (2_5) at (0,4.5){};
\node[circle,draw=black, fill=white, inner sep=0pt,minimum size=6pt] (2_6) at (0.5,4.5){};
\node[circle] (22) at (1,4.5){};
\draw (2) -- (2_0) -- (2_1) -- (2_2) -- (2_3) -- (2_4) -- (2_5) -- (2_6) -- (22);
\draw (-3.5,4.5) node{(b)};
\node[circle,fill=black, inner sep=0pt,minimum size=6pt] (2_a) at (-1.5,4.5){};
\draw [ultra thick] (2_4) -- (2_5);
\draw [ultra thick] (2) -- (2_0);

\node[circle] (3) at (-3,3.5){};
\node[circle,draw=black, fill=white, inner sep=0pt,minimum size=6pt] (3_0) at (-2.5,3.5){};
\node[circle,draw=black, fill=white, inner sep=0pt,minimum size=6pt] (3_1) at (-2,3.5){};
\node[circle,draw=black, fill=white, inner sep=0pt,minimum size=6pt] (3_2) at (-1.5,3.5){};
\node[circle,draw=black, fill=white, inner sep=0pt,minimum size=6pt] (3_3) at (-1,3.5){};
\node[circle,draw=black, fill=white, inner sep=0pt,minimum size=6pt] (3_4) at (-0.5,3.5){};
\node[circle,draw=black, fill=white, inner sep=0pt,minimum size=6pt] (3_5) at (0,3.5){};
\node[circle,draw=black, fill=white, inner sep=0pt,minimum size=6pt] (3_6) at (0.5,3.5){};
\node[circle] (33) at (1,3.5){};
\draw (3) -- (3_0) -- (3_1) -- (3_2) -- (3_3) -- (3_4) -- (3_5) -- (3_6) -- (33);
\draw (-3.5,3.5) node{(c)};
\node[circle,fill=black, inner sep=0pt,minimum size=6pt] (3_a) at (-1,3.5){};
\draw [ultra thick] (3_5) -- (3_6);
\draw [ultra thick] (3_0) -- (3_1);

\node[circle] (4) at (-3,2.5){};
\node[circle,draw=black, fill=white, inner sep=0pt,minimum size=6pt] (4_0) at (-2.5,2.5){};
\node[circle,draw=black, fill=white, inner sep=0pt,minimum size=6pt] (4_1) at (-2,2.5){};
\node[circle,draw=black, fill=white, inner sep=0pt,minimum size=6pt] (4_2) at (-1.5,2.5){};
\node[circle,draw=black, fill=white, inner sep=0pt,minimum size=6pt] (4_3) at (-1,2.5){};
\node[circle,draw=black, fill=white, inner sep=0pt,minimum size=6pt] (4_4) at (-0.5,2.5){};
\node[circle,draw=black, fill=white, inner sep=0pt,minimum size=6pt] (4_5) at (0,2.5){};
\node[circle,draw=black, fill=white, inner sep=0pt,minimum size=6pt] (4_6) at (0.5,2.5){};
\node[circle] (44) at (1,2.5){};
\draw (4) -- (4_0) -- (4_1) -- (4_2) -- (4_3) -- (4_4) -- (4_5) -- (4_6) -- (44);
\draw (-3.5,2.5) node{(d)};
\node[circle,fill=black, inner sep=0pt,minimum size=6pt] (4_a) at (-0.5,2.5){};
\draw [ultra thick] (4_6) -- (44);
\draw [ultra thick] (4_1) -- (4_2);

\node[circle] (5) at (-3,1.5){};
\node[circle,draw=black, fill=white, inner sep=0pt,minimum size=6pt] (5_0) at (-2.5,1.5){};
\node[circle,draw=black, fill=white, inner sep=0pt,minimum size=6pt] (5_1) at (-2,1.5){};
\node[circle,draw=black, fill=white, inner sep=0pt,minimum size=6pt] (5_2) at (-1.5,1.5){};
\node[circle,draw=black, fill=white, inner sep=0pt,minimum size=6pt] (5_3) at (-1,1.5){};
\node[circle,draw=black, fill=white, inner sep=0pt,minimum size=6pt] (5_4) at (-0.5,1.5){};
\node[circle,draw=black, fill=white, inner sep=0pt,minimum size=6pt] (5_5) at (0,1.5){};
\node[circle,draw=black, fill=white, inner sep=0pt,minimum size=6pt] (5_6) at (0.5,1.5){};
\node[circle] (55) at (1,1.5){};
\draw (5) -- (5_0) -- (5_1) -- (5_2) -- (5_3) -- (5_4) -- (5_5) -- (5_6) -- (55);
\draw (-3.5,1.5) node{(e)};
\node[circle,fill=black, inner sep=0pt,minimum size=6pt] (5_a) at (0,1.5){};
\node[circle,fill=black, inner sep=0pt,minimum size=6pt] (5_b) at (-2.5,1.5){};
\draw [ultra thick] (5_2) -- (5_3);

\draw [dashed] (-2.25, 1) -- (-2.25, 6);
\draw [dashed] (0.25, 1) -- (0.25, 6);

\end{tikzpicture}
\caption{Main part of the construction with the five possible configurations. Filled vertices are in $D$, thick edges are in $M$. }
\label{FigureMainPart}
\end{figure}
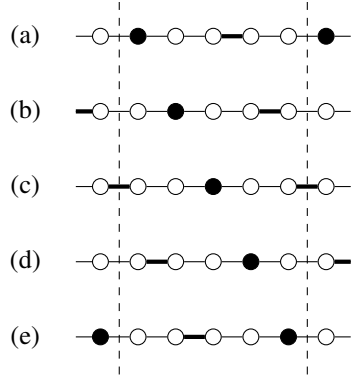

The main difficulty in completing the proof is showing that the optimal
solution has the desired form, and in particular, that the pattern that is
selected for a variable is kept constant throughout the construction.  This is
in general not possible to prove, but using a technique introduced by
\cite{LokshtanovMS18}, we work around this difficulty by making polynomially
many copies of our construction, gluing them together, and arguing that a large
enough consistent copy must exist.

\subsubsection*{Construction} We are given a $q$-CSP-$5$ instance $\varphi$
with $n$ variables $x_1,\ldots,x_n$ taking values over the set $\{0,1,2,3,4\}$,
and $m$ constraints $c_0,\ldots, c_{m-1}$.  For each constraint we are given a
set of at most $q$ variables which are involved in this constraint and a list
of satisfying assignments for these variables. Without loss of generality, we
make the following assumptions: (i) each constraint involves exactly $q$
variables, because if it has fewer variables, we can add to it new variables
and augment the list of satisfying assignments so that the value of the new
variables is irrelevant (ii) all constraints have lists of satisfying
assignments of size $C=5^q-1$; note that this is an  upper bound on the size of
the list of satisfying assignments since if a constraint has $5^q$ different satisfying assignments then it is always satisfied and thus is redundant; and for each constraint which has fewer we
add several copies of one of its satisfying assignments to its list (so the
list may repeat an assignment). We define two ``large'' numbers $F =
(4n+1)(2n+1)$ and $A = 12$ and we set our budget to be $k=8AFmn+
2Fmn+2Fmq(C-1)+n+1$.

We now construct our graph as follows:

\begin{enumerate}

\item We construct a vertex $s$ and attach to it two leaves $s_1,s_2$.

\item For $i\in\{1,\ldots,n\}$ we construct a path on $5Fm$ vertices: the
vertices are labeled $u_{i,j}$, for $j\in\{0,1,\ldots,5Fm-1\}$ and for each
$i,j$ the vertex $u_{i,j}$ is connected to $u_{i,j+1}$. We call these paths the
\emph{main} part of our construction.

\item For each $j\in\{0,1,\ldots,Fm-1\}$, let $j'=j\bmod m$.  We construct a
checker gadget $H_j$ as follows (see Figure \ref{FigureGadgetQijAndCopies}):

	\begin{enumerate}

	\item For each satisfying assignment $\sigma$ in the list of the
constraint $c_{j'}$, we construct an independent set $Z_{\sigma,j}$ of size
$2q$ (therefore, $C$ such independent sets). The $2q$ vertices are partitioned
so that for each of the $q$ variables involved in $c_{j'}$ we reserve two
vertices. In particular, if $x_i$ is involved in $c_{j'}$ we denote by
$z^1_{\sigma,j,i}, z^2_{\sigma,j,i}$ its two reserved vertices in
$Z_{\sigma,j}$.

		\item For each $i\in\{1,\ldots,n\}$ such that $x_i$ is involved
in $c_{j'}$, for each satisfying assignment $\sigma$ in the list of $c_{j'}$,
if $\sigma$ sets $x_i$ to value $\alpha\in\{0,1,2,3,4\}$ we add the following
edges:
		
	\begin{enumerate}

	\item $(u_{i,5j+\alpha}, z^1_{\sigma,j,i})$ and $(u_{i,5j+\alpha},
z^2_{\sigma,j,i})$.

	\item Let $\beta = (\alpha+2) \bmod 5$ and $\gamma = (\alpha+3)\bmod
5$. We add the edges  $(u_{i,5j+\beta}, z^1_{\sigma,j,i})$ and
$(u_{i,5j+\gamma}, z^2_{\sigma,j,i})$.

	\end{enumerate}

	\item For all assignments $\sigma\neq\sigma'$ of $c_{j'}$, add all
edges between $Z_{\sigma,j}$ and $Z_{\sigma',j}$.

	\item We construct an independent set $W_j$ of size $2q(C-1)$.

	\item Add all edges between $W_j$ and $Z_{\sigma,j}$, for all
assignments $\sigma$ of $c_{j'}$.

	\item For each $w\in W_j$, we construct an independent set of size
$2k+1$ whose vertices are all connected to $w$ and to $s$.

	\end{enumerate}

\item We define the consistency gadget $Q_{i,j}$, for $i\in\{1,\ldots,n\}$ and
$j\in\{0,\ldots,Fm-1\}$  which consists of (see Figure \ref{FigureGadgetQijAndCopies}):

	\begin{enumerate}

	\item An independent set of size $8$ denoted $A_{i,j}$.

	\item Five independent sets of size $2$ each, denoted $B_{i,j,0},
B_{i,j,1},\ldots, B_{i,j,4}$.

	\item For each $\ell,\ell' \in \{0,\ldots,4\}$ with $\ell\neq  \ell'$
all edges from $B_{i,j,\ell}$ to $B_{i,j,\ell'}$.

	\item For each $\ell \in \{0,\ldots,4\}$ all possible edges from
$B_{i,j,\ell}$ to $A_{i,j}$.

	\item For each $a\in A_{i,j}$, $2k+1$ vertices connected to $a$ and to
$s$.

	\item For each $\ell\in\{0,\ldots, 4\}$ both vertices of $B_{i,j,\ell}$
are connected to $u_{i,5j+\ell}$.

	\item For each $\ell\in\{0,\ldots, 4\}$ let $\ell'=(\ell+2)\bmod 5$ and
$\ell''=(\ell+3)\bmod 5$. One vertex of $B_{i,j,\ell}$ is connected to
$u_{i,5j+\ell'}$ and the other to $u_{i,5j+\ell''}$.  

	\end{enumerate}

\item For each $i\in\{1,\ldots,n\}$ and $j\in\{0,\ldots,Fm-1\}$ construct $A$
copies of the gadget $Q_{i,j}$ and connect them to the main part as described
above.

\end{enumerate}

This completes the construction. 

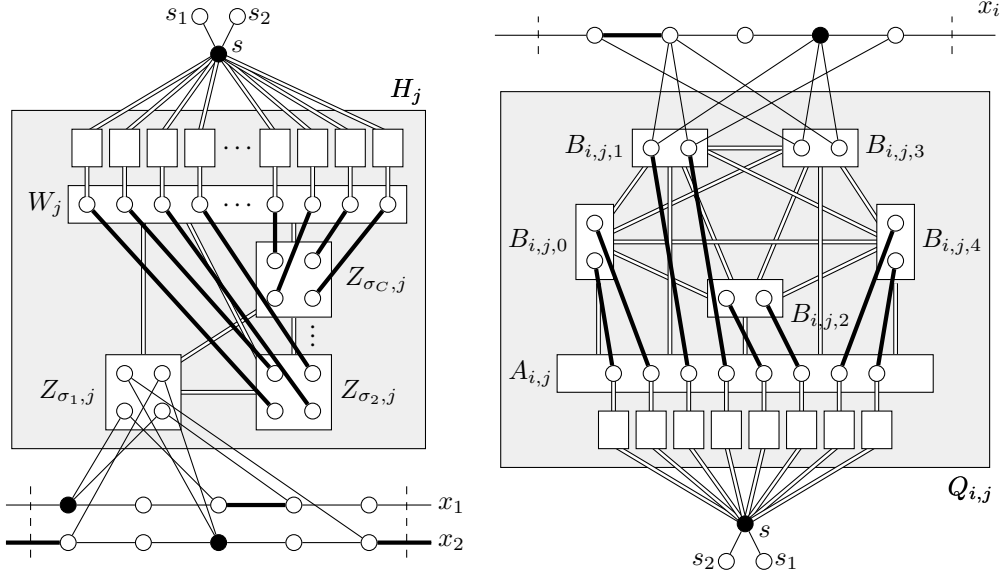
\begin{figure}[h]
\centering
\begin{tikzpicture}

\node[circle,draw=black, fill=white, inner sep=0pt,minimum size=6pt] (4_0) at (4,12.25){};
\node[circle,draw=black, fill=white, inner sep=0pt,minimum size=6pt] (4_1) at (5,12.25){};
\node[circle,draw=black, fill=white, inner sep=0pt,minimum size=6pt] (4_2) at (6,12.25){};
\node[circle,draw=black, fill=black, inner sep=0pt,minimum size=6pt] (4_3) at (7,12.25){};
\node[circle,draw=black, fill=white, inner sep=0pt,minimum size=6pt] (4_4) at (8,12.25){};
\node[circle] (4_0_0) at (2.5,12.25){};
\node[circle] (4_0_1) at (9.5,12.25){};
\draw (4_0_0) -- (4_0) -- (4_1) -- (4_2) -- (4_3) -- (4_4) -- (4_0_1);
\draw (9.25, 12.6) node{$x_i$};

\draw [ultra thick] (4_0) -- (4_1);

\draw [fill = lightgray!25] (2.75,6.5) rectangle (9.25,11.5);
\draw (9,6.2) node{$Q_{i,j}$};
\draw (9,6.2) node{$Q_{i,j}$};

\draw [dashed] (3.25, 12) -- (3.25, 12.5);
\draw [dashed] (8.75, 12) -- (8.75, 12.5);

\draw [double distance = 1pt] (6,8.55) -- (3.95,9.5);
\draw [double distance = 1pt] (6,8.55) -- (7.95,9.5);
\draw [double distance = 1pt] (6,8.55) -- (5,10.95);
\draw [double distance = 1pt] (6,8.55) -- (7,10.95);
\draw [double distance = 1pt] (3.95,9.5) -- (7.95,9.5);
\draw [double distance = 1pt] (3.95,9.5) -- (5,10.95);
\draw [double distance = 1pt] (3.95,9.5) -- (7,10.95);
\draw [double distance = 1pt] (5,10.95) -- (7.95,9.5);
\draw [double distance = 1pt] (5.45,10.75) -- (6.55,10.75);
\draw [double distance = 1pt] (7,10.95) -- (7.95,9.5);

\draw [double distance = 1pt] (6,8.5) -- (6,7.95);
\draw [double distance = 1pt] (4.05,8.95) -- (4.05,7.95);
\draw [double distance = 1pt] (8,8.95) -- (8,7.95);
\draw [double distance = 1pt] (5,10.55) -- (5,7.95);
\draw [double distance = 1pt] (7,10.55) -- (7,7.95);

\draw [fill = white] (5.5,8.5) rectangle (6.5,9);
\node[circle,draw=black, fill=white, inner sep=0pt,minimum size=6pt] (a_3) at (5.75, 8.75){};
\node[circle,draw=black, fill=white, inner sep=0pt,minimum size=6pt] (a_4) at (6.25, 8.75){};
\draw (7, 8.5) node{$B_{i,j,2}$};

\draw [fill = white] (3.75,9) rectangle (4.25,10);
\node[circle,draw=black, fill=white, inner sep=0pt,minimum size=6pt] (b_3) at (4, 9.25){};
\node[circle,draw=black, fill=white, inner sep=0pt,minimum size=6pt] (b_4) at (4, 9.75){};
\draw (3.25, 9.5) node{$B_{i,j,0}$};

\draw [fill = white] (7.75,9) rectangle (8.25,10);
\node[circle,draw=black, fill=white, inner sep=0pt,minimum size=6pt] (c_3) at (8, 9.25){};
\node[circle,draw=black, fill=white, inner sep=0pt,minimum size=6pt] (c_4) at (8, 9.75){};
\draw (8.75, 9.5) node{$B_{i,j,4}$};

\draw [fill = white] (4.5,10.5) rectangle (5.5,11);
\node[circle,draw=black, fill=white, inner sep=0pt,minimum size=6pt] (d_3) at (4.75, 10.75){};
\node[circle,draw=black, fill=white, inner sep=0pt,minimum size=6pt] (d_4) at (5.25, 10.75){};
\draw (4, 10.75) node{$B_{i,j,1}$};

\draw [fill = white] (6.5,10.5) rectangle (7.5,11);
\node[circle,draw=black, fill=white, inner sep=0pt,minimum size=6pt] (e_3) at (6.75, 10.75){};
\node[circle,draw=black, fill=white, inner sep=0pt,minimum size=6pt] (e_4) at (7.25, 10.75){};
\draw (8, 10.75) node{$B_{i,j,3}$};

\draw [fill = white] (3.5,7.5) rectangle (8.5,8);

\draw [double distance = 1pt] (4.25,6.8) -- (6,5.75);
\draw [double distance = 1pt] (4.75,6.8) -- (6,5.75);
\draw [double distance = 1pt] (5.25,6.8) -- (6,5.75);
\draw [double distance = 1pt] (5.75,6.8) -- (6,5.75);
\draw [double distance = 1pt] (6.25,6.8) -- (6,5.75);
\draw [double distance = 1pt] (6.75,6.8) -- (6,5.75);
\draw [double distance = 1pt] (7.25,6.8) -- (6,5.75);
\draw [double distance = 1pt] (7.75,6.8) -- (6,5.75);

\draw [double distance = 1pt] (4.25,7.65) -- (4.25,7.15);
\draw [double distance = 1pt] (4.75,7.65) -- (4.75,7.15);
\draw [double distance = 1pt] (5.25,7.65) -- (5.25,7.15);
\draw [double distance = 1pt] (5.75,7.65) -- (5.75,7.15);
\draw [double distance = 1pt] (6.25,7.65) -- (6.25,7.15);
\draw [double distance = 1pt] (6.75,7.65) -- (6.75,7.15);
\draw [double distance = 1pt] (7.25,7.65) -- (7.25,7.15);
\draw [double distance = 1pt] (7.75,7.65) -- (7.75,7.15);

\node[circle,draw=black, fill=white, inner sep=0pt,minimum size=6pt] (w_1) at (4.25,7.75){};
\node[circle,draw=black, fill=white, inner sep=0pt,minimum size=6pt] (w_2) at (4.75,7.75){};
\node[circle,draw=black, fill=white, inner sep=0pt,minimum size=6pt] (w_3) at (5.25,7.75){};
\node[circle,draw=black, fill=white, inner sep=0pt,minimum size=6pt] (w_4) at (5.75,7.75){};
\node[circle,draw=black, fill=white, inner sep=0pt,minimum size=6pt] (w_5) at (6.25,7.75){};
\node[circle,draw=black, fill=white, inner sep=0pt,minimum size=6pt] (w_6) at (6.75,7.75){};
\node[circle,draw=black, fill=white, inner sep=0pt,minimum size=6pt] (w_7) at (7.25,7.75){};
\node[circle,draw=black, fill=white, inner sep=0pt,minimum size=6pt] (w_8) at (7.75,7.75){};
\draw (3.15, 7.75) node{$A_{i,j}$};

\draw [fill = white] (4.05,7.25) rectangle (4.45,6.75);
\draw [fill = white] (4.55,7.25) rectangle (4.95,6.75);
\draw [fill = white] (5.05,7.25) rectangle (5.45,6.75);
\draw [fill = white] (5.55,7.25) rectangle (5.95,6.75);
\draw [fill = white] (6.05,7.25) rectangle (6.45,6.75);
\draw [fill = white] (6.55,7.25) rectangle (6.95,6.75);
\draw [fill = white] (7.05,7.25) rectangle (7.45,6.75);
\draw [fill = white] (7.55,7.25) rectangle (7.95,6.75);

\draw (4_1) -- (d_3);
\draw (4_3) -- (d_3);
\draw (4_4) -- (d_4);
\draw (4_1) -- (d_4);
\draw (4_3) -- (e_3);
\draw (4_0) -- (e_3);
\draw (4_3) -- (e_4);
\draw (4_1) -- (e_4);

\draw [ultra thick] (b_3) -- (w_1);
\draw [ultra thick] (b_4) -- (w_2);
\draw [ultra thick] (d_3) -- (w_3);
\draw [ultra thick] (d_4) -- (w_4);
\draw [ultra thick] (a_3) -- (w_5);
\draw [ultra thick] (a_4) -- (w_6);
\draw [ultra thick] (c_4) -- (w_7);
\draw [ultra thick] (c_3) -- (w_8);

\node[circle,draw=black, fill=black, inner sep=0pt,minimum size=6pt] (s) at (6,5.75){};
\draw (6.26, 5.65) node{$s$};
\node[circle,draw=black, fill=white, inner sep=0pt,minimum size=6pt] (s_1) at (6.25,5.25){};
\draw (6.55, 5.25) node{$s_1$};
\node[circle,draw=black, fill=white, inner sep=0pt,minimum size=6pt] (s_2) at (5.75,5.25){};
\draw (5.45, 5.25) node{$s_2$};
\draw (s) -- (s_1);
\draw (s) -- (s_2);

\node[circle,draw=black, fill=black, inner sep=0pt,minimum size=6pt] (4_0) at (-3,6){};
\node[circle,draw=black, fill=white, inner sep=0pt,minimum size=6pt] (4_1) at (-2,6){};
\node[circle,draw=black, fill=white, inner sep=0pt,minimum size=6pt] (4_2) at (-1,6){};
\node[circle,draw=black, fill=white, inner sep=0pt,minimum size=6pt] (4_3) at (0,6){};
\node[circle,draw=black, fill=white, inner sep=0pt,minimum size=6pt] (4_4) at (1,6){};
\node[circle, minimum size=0pt] (4_0_0) at (-4,6){};
\node[circle, minimum size=0pt] (4_0_1) at (2,6){};
\draw (4_0_0) -- (4_0) -- (4_1) -- (4_2) -- (4_3) -- (4_4) -- (4_0_1);

\node[circle,draw=black, fill=white, inner sep=0pt,minimum size=6pt] (5_0) at (-3,5.5){};
\node[circle,draw=black, fill=white, inner sep=0pt,minimum size=6pt] (5_1) at (-2,5.5){};
\node[circle,draw=black, fill=black, inner sep=0pt,minimum size=6pt] (5_2) at (-1,5.5){};
\node[circle,draw=black, fill=white, inner sep=0pt,minimum size=6pt] (5_3) at (0,5.5){};
\node[circle,draw=black, fill=white, inner sep=0pt,minimum size=6pt] (5_4) at (1,5.5){};
\node[circle, minimum size=0pt] (5_0_0) at (-4,5.5){};
\node[circle, minimum size=0pt] (5_0_1) at (2,5.5){};
\draw (5_0_0) -- (5_0) -- (5_1) -- (5_2) -- (5_3) -- (5_4) -- (5_0_1);

\draw [fill = lightgray!25] (1.75,6.75) rectangle (-3.75,11.25);
\draw (1.5,11.5) node{$H_{j}$};
\draw (1.5,11.5) node{$H_{j}$};

\draw [double distance = 1pt] (-1.7,7.8) -- (-0.3,8.7);
\draw [double distance = 1pt] (-1.5,7.5) -- (-0.5,7.5);
\draw [double distance = 1pt] (0,8) -- (0,8.5);

\draw [double distance = 1pt] (-2,7.8) -- (-2,9.95);
\draw [double distance = 1pt] (-0.4,7.8) -- (-1.5,9.95);
\draw [double distance = 1pt] (0,9.3) -- (0,9.95);

\draw [fill = white] (-2.5,7) rectangle (-1.5,8);
\node[circle,draw=black, fill=white, inner sep=0pt,minimum size=6pt] (a_1) at (-2.25, 7.25){};
\node[circle,draw=black, fill=white, inner sep=0pt,minimum size=6pt] (a_2) at (-1.75, 7.25){};
\node[circle,draw=black, fill=white, inner sep=0pt,minimum size=6pt] (a_3) at (-2.25, 7.75){};
\node[circle,draw=black, fill=white, inner sep=0pt,minimum size=6pt] (a_4) at (-1.75, 7.75){};
\draw (-3, 7.5) node{$Z_{\sigma_1,j}$};

\draw [fill = white] (-0.5,7) rectangle (0.5,8);
\node[circle,draw=black, fill=white, inner sep=0pt,minimum size=6pt] (b_1) at (-0.25, 7.25){};
\node[circle,draw=black, fill=white, inner sep=0pt,minimum size=6pt] (b_2) at (0.25, 7.25){};
\node[circle,draw=black, fill=white, inner sep=0pt,minimum size=6pt] (b_3) at (-0.25, 7.75){};
\node[circle,draw=black, fill=white, inner sep=0pt,minimum size=6pt] (b_4) at (0.25, 7.75){};
\draw (1, 7.5) node{$Z_{\sigma_2,j}$};

\draw (4_0) -- (a_1);
\draw (4_2) -- (a_1);
\draw (4_0) -- (a_2);
\draw (4_3) -- (a_2);
\draw (5_2) -- (a_3);
\draw (5_4) -- (a_3);
\draw (5_2) -- (a_4);
\draw (5_0) -- (a_4);

\draw (0.25, 8.125) node{.};
\draw (0.25, 8.25) node{.};
\draw (0.25, 8.375) node{.};

\draw [fill = white] (-0.5,8.5) rectangle (0.5,9.5);
\node[circle,draw=black, fill=white, inner sep=0pt,minimum size=6pt] (c_1) at (-0.25, 8.75){};
\node[circle,draw=black, fill=white, inner sep=0pt,minimum size=6pt] (c_2) at (0.25, 8.75){};
\node[circle,draw=black, fill=white, inner sep=0pt,minimum size=6pt] (c_3) at (-0.25, 9.25){};
\node[circle,draw=black, fill=white, inner sep=0pt,minimum size=6pt] (c_4) at (0.25, 9.25){};
\draw (1.1, 9) node{$Z_{\sigma_{C},j}$};

\draw [fill = white] (-3,9.75) rectangle (1.5,10.25);

\node[circle,draw=black, fill=white, inner sep=0pt,minimum size=6pt] (w_1) at (-2.75,10){};
\node[circle,draw=black, fill=white, inner sep=0pt,minimum size=6pt] (w_2) at (-2.25,10){};
\node[circle,draw=black, fill=white, inner sep=0pt,minimum size=6pt] (w_3) at (-1.75,10){};
\node[circle,draw=black, fill=white, inner sep=0pt,minimum size=6pt] (w_4) at (-1.25,10){};
\draw (-0.9, 10) node{.};
\draw (-0.75, 10) node{.};
\draw (-0.6, 10) node{.};
\node[circle,draw=black, fill=white, inner sep=0pt,minimum size=6pt] (w_6) at (-0.25,10){};
\node[circle,draw=black, fill=white, inner sep=0pt,minimum size=6pt] (w_7) at (0.25,10){};
\node[circle,draw=black, fill=white, inner sep=0pt,minimum size=6pt] (w_8) at (0.75,10){};
\node[circle,draw=black, fill=white, inner sep=0pt,minimum size=6pt] (w_9) at (1.25,10){};
\draw (-3.3, 10) node{$W_j$};

\draw [ultra thick] (b_1) -- (w_1);
\draw [ultra thick] (b_3) -- (w_2);
\draw [ultra thick] (b_2) -- (w_3);
\draw [ultra thick] (b_4) -- (w_4);
\draw [ultra thick] (c_3) -- (w_6);
\draw [ultra thick] (c_1) -- (w_7);
\draw [ultra thick] (c_4) -- (w_8);
\draw [ultra thick] (c_2) -- (w_9);

\draw [double distance = 1pt] (-2.75,10.95) -- (-1,12);
\draw [double distance = 1pt] (-2.25,10.95) -- (-1,12);
\draw [double distance = 1pt] (-1.75,10.95) -- (-1,12);
\draw [double distance = 1pt] (-1.25,10.95) -- (-1,12);
\draw [double distance = 1pt] (-0.25,10.95) -- (-1,12);
\draw [double distance = 1pt] (0.25,10.95) -- (-1,12);
\draw [double distance = 1pt] (0.75,10.95) -- (-1,12);
\draw [double distance = 1pt] (1.25,10.95) -- (-1,12);

\draw [fill = white] (-2.95,10.5) rectangle (-2.55,11);
\draw [fill = white] (-2.45,10.5) rectangle (-2.05,11);
\draw [fill = white] (-1.95,10.5) rectangle (-1.55,11);
\draw [fill = white] (-1.45,10.5) rectangle (-1.05,11);
\draw [fill = white] (-0.45,10.5) rectangle (-0.05,11);
\draw [fill = white] (0.05,10.5) rectangle (0.45,11);
\draw [fill = white] (0.55,10.5) rectangle (0.95,11);
\draw [fill = white] (1.05,10.5) rectangle (1.45,11);

\draw (-0.9, 10.75) node{.};
\draw (-0.75, 10.75) node{.};
\draw (-0.6, 10.75) node{.};

\draw [double distance = 1pt] (-2.75,10.1) -- (-2.75,10.5);
\draw [double distance = 1pt] (-2.25,10.1) -- (-2.25,10.5);
\draw [double distance = 1pt] (-1.75,10.1) -- (-1.75,10.5);
\draw [double distance = 1pt] (-1.25,10.1) -- (-1.25,10.5);
\draw [double distance = 1pt] (-0.25,10.1) -- (-0.25,10.5);
\draw [double distance = 1pt] (0.25,10.1) -- (0.25,10.5);
\draw [double distance = 1pt] (0.75,10.1) -- (0.75,10.5);
\draw [double distance = 1pt] (1.25,10.1) -- (1.25,10.5);

\node[circle,draw=black, fill=black, inner sep=0pt,minimum size=6pt] (s) at (-1,12){};
\draw (-0.75, 12.1) node{$s$};
\node[circle,draw=black, fill=white, inner sep=0pt,minimum size=6pt] (s_1) at (-1.25,12.5){};
\draw (-1.55, 12.5) node{$s_1$};
\node[circle,draw=black, fill=white, inner sep=0pt,minimum size=6pt] (s_2) at (-0.75,12.5){};
\draw (-0.45, 12.5) node{$s_2$};
\draw (s) -- (s_1);
\draw (s) -- (s_2);

\draw [dashed] (-3.5, 5.3) -- (-3.5, 6.3);
\draw [dashed] (1.5, 5.3) -- (1.5, 6.3);

\draw (2.1, 5.5) node{$x_2$};
\draw (2.1, 6) node{$x_1$};

\draw [ultra thick] (4_2) -- (4_3);
\draw [ultra thick] (5_0_0) -- (5_0);
\draw [ultra thick] (5_4) -- (5_0_1);

\end{tikzpicture}
\caption{(Double edges between two sets of vertices represent all edges between the two sets.) 
Left: Checker gadget $H_j$ connected to the main part. Here we have considered an instance where the clause $c_{j'}$ has only two variables, $x_1$ and $x_2$. Moreover, only the independent set $Z_{\sigma_1,j}$ is shown connected to the main part. The possible assignment $\sigma_1$ of $c_{j'}$ is $(x_1 = 0, x_2 = 2)$. We have supposed that this assignment is satisfiable, and we have marked the corresponding mixed dominating set: filled vertices are in $D$, thick edges are in $M$. 
Right: Checker gadget $Q_{i,j}$ connected to the main part, that is to the path corresponding to the variable $x_i$. Only the independent sets $B_{i,j,1}$ and $B_{i,j,3}$ are shown connected to the main part. We have supposed that the assignment $(x_i = 3)$ is satisfiable, and we have marked the corresponding mixed dominating set: filled vertices are in $D$, thick edges are in $M$. }
\label{FigureGadgetQijAndCopies}
\end{figure}

The target size is $k$, as defined above.
We now argue that the reduction is correct and $G$ has the desired pathwidth.

\begin{lemma}\label{lemmacsptomds} If $\varphi$ is satisfiable, then there
exists a mixed dominating set in $G$ of size at most $k$.  \end{lemma}

\begin{proof} Assume that $\varphi$ admits some satisfying assignment $\rho:\{x_1,\ldots,
x_n\}\to \{0,1,2,3,4\}$. We construct a solution as follows:

\begin{enumerate}

\item For each $i\in\{1,\ldots,n\}$ let $\alpha=\rho(x_i)$. For each
$j\in\{0,\ldots,Fm-1\}$, we select in the dominating set the vertex $u_{i,5j+\alpha}$.

\item Let $U'$ be the set of vertices $u_{i,j}$ of the main part  which were not
selected in the previous step and which do not have a neighbor selected in the
previous step.  We add to the solution all edges of a maximum matching of
$G[U']$, as well as all vertices of $U'$ left unmatched by this matching.

\item For each $j\in\{0,\ldots,Fm-1\}$, $G$ contains a gadget $H_j$.  Consider
the constraint $c_{j'}$ for $j'=j\bmod m$.  Let $\sigma$ be an assignment in
the list of $c_{j'}$ that agrees with $\rho$ (such a $\sigma$ must exist, since
the constraint is satisfied by $\rho$). We add to the solution the edges of a
perfect matching from $W_j$ to $\bigcup_{\sigma'\neq \sigma} Z_{\sigma',j}$.

\item For each $j\in\{0,\ldots,Fm-1\}$ and $i\in\{1,\ldots,n\}$ we have added
to the graph $A$ copies of the consistency gadget $Q_{i,j}$. For each copy we
add to the solution a perfect matching from $A_{i,j}$ to
$\bigcup_{\ell\neq\rho(x_i)} B_{i,j,\ell}$.

\item We set $s\in D$.

\end{enumerate}

Let us first argue why this solution has size at most $k$. In the first step we
select $Fnm$ vertices. In the second step we select at most $Fnm+n$ elements.
To see this, note that if $u_{i,j}$ is taken in the previous step, then
$u_{i,j+5}$ is also taken (assuming $j+5<5Fm$), which leaves two adjacent
vertices ($u_{i,j+2}, u_{i,j+3}$). These vertices will be matched in $G[U']$
and in our solution. Note that, for a variable $x_i$, if $\rho(x_i) \neq 2$, then at most one vertex is left unmatched by the matching taken, so the cost for this variable is at most $Fm+1$. If $\rho(x_i) = 2$, then at most two vertices are left matched by the matching taken, so the cost for this variable is at most $(Fm-1)+2$. 
Furthermore, for each $H_j$ we select $|W_j| = 2q(C-1)$
edges. For each copy of $Q_{i,j}$ we select $8$ edges, for a total cost of
$8AFmn$. Taking into account $s$, the total cost is at most
$Fnm+Fnm+n+2Fmq(C-1)+8AFmn+1=k$.

Let us argue why the solution is feasible. First, all vertices $u_{i,j}$ and all edges connecting them to each other are dominated by the first two steps of our selection since we have taken a maximum matching in $G[U']$ and all vertices left unmatched by this matching.
Second, for each $H_j$, the vertex $s$ together with
the endpoints of selected edges form a vertex cover of $H_j$, so all internal
edges are dominated. Furthermore, $s$ dominates all vertices which are not
endpoints of our solution, except $Z_{\sigma,j}$, where $\sigma$ is the
selected assignment of $c_{j'}$, with $j'=j\bmod m$. We then need to argue that
the vertices of $Z_{\sigma,j}$ and the edges connecting it to the main part are
covered.

Recall that the $2q$ vertices of $Z_{\sigma,j}$ are partitioned into pairs,
with each pair $z^1_{\sigma,j,i}, z^2_{\sigma,j,i}$ reserved for the variable
$x_i$ involved in $c_{j'}$. We now claim that $z^1_{\sigma,j,i},
z^2_{\sigma,j,i}$ are dominated by our solution, since we have selected the
vertex $u_{i,5j+\alpha}$, where $\alpha=\rho(x_i)$. Furthermore,
$u_{i,5j+\beta}, u_{i,5j+\gamma}$, where $\beta=(a+2)\bmod m$,
$\gamma=(a+3)\bmod m$, belong in $U'$ and therefore the edges incident to them
are covered.  Finally, to see that the $Q_{i,j}$ gadgets are covered, observe
that for each such gadget only $2$ vertices of some $B_{i,j,\ell}$ are not in
$P$. The common neighbor of these vertices is in $D$, and their other neighbors
in the main part are in $P$.  
\end{proof}

The idea of the proof of the next Lemma is the following: by partitioning the graph into different parts and lower bound the cost of these parts, we prove that if a mixed dominating set in $G$ has not the same form as in Lemma \ref{lemmacsptomds} in a sufficiently large copy, then it has size strictly greater than $k$, enabling us to produce a satisfiable assignment for $\varphi$ using the mixed dominating set which has the desired form. 

\begin{lemma}\label{lemmamdstocsp} If there exists a mixed dominating set in
$G$ of size at most $k$, then $\varphi$ is satisfiable.  \end{lemma}

\begin{proof}
Suppose that we are given, without loss of generality (\Cref{lemmanice}), a
nice mixed dominating set of $G$ of minimum cost. 
We therefore have a partition of $V$ into $V = D \cup P \cup I$.
Before proceeding, let us define for a set $S\subseteq V$ its \emph{cost} as
$\cost(S)=|S\cap D|+\frac{|S\cap P|}{2}$. 
Clearly, $\cost(V) \leq k$ since $|M| \geq |P|/2$, and 
for
disjoint sets $S_1,S_2$ we have $\cost(S_1\cup S_2) = \cost(S_1)+\cost(S_2)$.
Our strategy will therefore be to partition $V$ into different parts and lower
bound their cost.

First, we give some notation. Consider some $j \in \{ 0, \ldots, Fm-1 \}$ and
$i \in \{ 1, \ldots, n \}$: recall that we have constructed $A$ copies of the
gadget $Q_{i,j}$, call them $Q^1_{i,j},\ldots, Q^A_{i,j}$; also we define the
sets $S_{i,j}=\{u_{i,5j}, u_{i,5j+1},\ldots, u_{i,5j+4}\}$. Now, for some $j
\in \{ 0, \ldots, Fm-1 \}$, let: \begin{equation}\label{equation:1} S_j = H_j
\cup \bigcup_{i\in\{1,\ldots, n\}}\left( S_{i,j} \cup
\bigcup_{r\in\{1,\ldots,A\}} Q^r_{i,j}\right)   \end{equation}

\begin{claim}
$\cost(S_j) \geq 2q(C-1)+2n + 8An$. 
\end{claim}

\begin{proof}
We begin with some easy observations. First, it must be the case that $s\in D$.
If not, either $s_1$ or $s_2$ are in $D$, which contradicts the niceness of the
solution. 

Consider some $j\in\{0,\ldots,Fm-1\}$ and $i\in\{1,\ldots,n\}$. 
We will say that, for $1 \leq r \leq A$, $Q^r_{i,j}$ is \emph{normal} if
we have the following: $Q^r_{i,j}\cap D=\emptyset$ and there exists $\ell\in
\{0,\ldots,4\}$ such that $Q^r_{i,j}\cap P = A_{i,j}\cup
\bigcup_{\ell'\neq\ell} B_{i,j,\ell'}$. In other words, $Q^r_{i,j}$ is normal if
locally the solution has the form described in \Cref{lemmacsptomds}.

We now observe that for all $i,j,r$ we have $\cost(Q^r_{i,j})\ge 8$. To see
this, observe that if there exists $a\in A_{i,j}\cap I$, then the $2k+1$
neighbors of $a$ must be in $D\cup P$, so the solution cannot have cost $k$.
Hence, $A_{i,j}\subseteq D\cup P$. Furthermore, the maximum independent set of
$\bigcup_{\ell\in\{0,\ldots,4\}} B_{i,j,\ell}$ is $2$, so
$|(\bigcup_{\ell\in\{0,\ldots,4\}} B_{i,j,\ell})\cap (D\cup P)|\ge 8 $.
Following this reasoning we also observe that if $Q^r_{i,j}$ is not normal,
then we have $\cost(Q^r_{i,j})>8$. In other words, $8$ is a lower bound for the
cost of every copy of $Q_{i,j}$, which can only be attained if a copy is
normal.

Consider some $j\in\{0,\ldots,Fm-1\}$ and $i\in\{1,\ldots,n\}$ and suppose that
none of the $A$ copies of $Q_{i,j}$ is normal. We will then arrive at a
contradiction. Indeed, we have $\cost(\bigcup_r Q^r_{i,j}) \ge 8A + A/2 \ge 8A
+ 6$. We create another solution by doing the following: take the five vertices $u_{i,5j}, u_{i,5j+1},\ldots, u_{i,5j+4}$, and take in all $Q_{i,j}$ a matching so that $Q_{i,j}$ is normal. This has decreased the total cost, while keeping the solution valid, which should not be possible.

We can therefore assume from now on that for each $i,j$ at least one copy of
$Q_{i,j}$ is normal, hence, there exists $\ell\in\{0,\ldots,4\}$ such that
$B_{i,j,\ell}\subseteq I$ in that copy. 

Recall that $S_{i,j}=\{u_{i,5j}, u_{i,5j+1},\ldots, u_{i,5j+4}\}$. We claim that for
all $i\in\{1,\ldots,n\}, j\in\{0,\ldots,Fm-1\}$, we have $\cost(S_{i,j})\ge 2$.
Indeed, if we consider the normal copy of $Q_{i,j}$ which has
$B_{i,j,\ell}\subseteq I$, the two vertices of $B_{i,j,\ell}$ have three
neighbors in $S_{i,j}$, and at least one of them must be in $D$ to dominate the vertices of $B_{i,j,\ell}$.

In addition, we claim that for all $j\in\{0,\ldots,Fm-1\}$ we have
$\cost(H_j)\ge 2q(C-1)$. The reasoning here is similar to $Q_{i,j}$, namely,
the vertices of $W_j$ cannot belong to $I$ (otherwise we get $2k+1$ vertices in
$D\cup P$) ; and from the $2qC$ vertices in $\bigcup_{\sigma} Z_{\sigma,j}$ at
most $2q$ can belong to $I$.

We now have the lower bounds we need: $\cost(S_j)\ge
2q(C-1)+2n+8An$. 
\end{proof}

Now, if for some $j$ we have  $\cost(S_j)> 2q(C-1)+2n+8An$ we will
say that $j$ is \emph{problematic}. 

\begin{claim}
There exists a contiguous interval $J \subseteq \{ 0, \ldots, Fm-1 \}$ of size at least $m(4n+1)$ in which all $j \in J$ are not problematic. 
\end{claim}

\begin{proof} Let $L\subseteq \{0,\ldots,Fm-1\}$ be the set of problematic indices. We claim
that $|L|\le 2n$.  Indeed, we have $\cost(V) =
1+\sum_{j\in\{0,\ldots,Fm-1\}} \cost(S_j) \ge 1 + Fm(2q(C-1)+2n+8An) + |L|/2 =
k-n+|L|/2$. But since the total cost is at most $k$, we have $|L|/2\le n$.

We will now consider the longest contiguous interval $J\subseteq \{0,\ldots,
Fm-1\}$ such that all $j\in J$ are not problematic. We have $|J|\ge
Fm/(|L|+1)\ge m(4n+1)$. 
\end{proof}

Before we proceed further, we note that if $j$ is not problematic, then for any
$i\in\{1,\ldots,n\}$, all edges of $M$ which have an endpoint in $S_{i,j}$,
must have their other endpoint also in the main part, that is, they must be
edges of the main paths.  To see this note that if $j$ is not problematic, all
$Q_{i,j}$ are normal, so there are $8$ vertices in $A_{i,j}\cap P$ which must
be matched to the $8$ vertices of $(\bigcup_{\ell} B_{i,j,\ell} )\cap P$.
Similarly, in $H_j$ the $2q(C-1)$ vertices of $W_j\cap P$ must be matched to
the $2q(C-1)$ vertices of $(\bigcup_{\sigma} Z_{\sigma,j}) \cap P$, otherwise we would
increase the cost and $j$ would be problematic.

Consider now a non-problematic $j \in J$ and $i \in \{ 1, \ldots, n \}$
such that $\cost(S_{i,j}) = 2$. We claim that the solution must follow one of
the five configurations below (see also Figure \ref{FigureMainPart}):

\begin{itemize}

\item[(a)] $u_{i,5j}\in D$ and $(u_{i,5j+2}, u_{i,5j+3})\in M$.

\item[(b)] $u_{i,5j+1}\in D$ and $(u_{i,5j+3}, u_{i,5j+4})\in M$.

\item[(c)] $u_{i,5j+2}\in D$,  $(u_{i,5j+4}, u_{i,5j+5})\in M$, and
$(u_{i,5j-1}, u_{i,5j})\in M$.

\item[(d)] $u_{i,5j+3}\in D$ and $(u_{i,5j}, u_{i,5j+1})\in M$.

\item[(e)] $u_{i,5j+4}\in D$ and $(u_{i,5j+1}, u_{i,5j+2})\in M$.

\end{itemize}

Indeed, these configurations cover all the cases where exactly one
vertex of $S_{i,j}$ is in $D$ and exactly two are in $P$. This is a condition
enforced by the fact that all of the $Q_{i,j}$ copies are normal, and that
$\cost(S_{i,j})=2$.

\begin{claim}
There exists a contiguous interval $J' \subseteq J$ of size at least $m$ in which all $j \in J'$ are not problematic and for all $j_1, j_2 \in J'$, $S_{i,j_1}$ and $S_{i,j_2}$ are in the same configuration. 
\end{claim}

\begin{proof} Given the five configurations shown in Figure \ref{FigureMainPart}, we now make the following simple observations, where statements apply for all
$i\in\{1,\ldots,n\}$ and $j$ such that $j,j+1\in J$:

\begin{itemize}
    \item If $S_{i,j}$ is in configuration (a), then $S_{i,j+1}$ is also in configuration (a). We note that in configuration (a) vertex $u_{i,5j+4}$ is not dominated, so $S_{i,j+1}$ cannot be in configuration (b), (d) and (e) by this fact, and cannot be in configuration (c) because otherwise $\cost(S_{i,j}) > 2$. 
    
    \item If $S_{i,j}$ is in configuration (b), then $S_{i,j+1}$ is in configuration (a), (b), (c) or (d). We note that $S_{i,j+1}$ cannot be in configuration (e) because otherwise the vertex $u_{i,5(j+1)}$ is not dominated.
    
    \item If $S_{i,j}$ is in configuration (c), then $S_{i,j+1}$ is in configuration (c) or (d). We note that $S_{i,j+1}$ cannot be in configuration (a) since $D \cap P = \emptyset$, nor in configuration (b) and (e) because otherwise $\cost(S_{i,j+1}) > 2$. 
    
    \item If $S_{i,j}$ is in configuration (d), then $S_{i,j+1}$ is in configuration (a) or (d). We note that $S_{i,j+1}$ cannot be in configuration (b) and (e) because otherwise the edge $(u_{i,5j+4}, u_{i,5(j+1)})$ is not dominated, nor in configuration (c) because otherwise $\cost(S_{i,j}) > 2$. 
    
    \item If $S_{i,j}$ is in configuration (e), then $S_{i,j+1}$ is in configuration (a), (b), (d) or (e). We note that $S_{i,j+1}$ cannot be in configuration (c) since $D \cap P = \emptyset$. 
\end{itemize}

We will now say for some $i \in \{ 1, \ldots, n \}$, $j \in J$, that $j$ is \emph{shifted} for variable $i$ if $j+1 \in J$ but $S_{i,j}$ and $S_{i,j+1}$ do not have the same configuration. We observe that there cannot exist distinct $j_1, j_2, j_3, j_4, j_5 \in J$ such that all of them are shifted for variable $i$. Indeed, if we draw a directed graph with a vertex for each configuration, and an arc $(u,v)$ expressing the property that the configuration represented by $v$ can follow the one represented by $u$, if we take into account the observations above, the graph will be a DAG with maximum path length $4$. Hence, a configuration cannot shift $5$ times, as long as we stay in $J$ (the part of the graph where the minimum local cost is attained everywhere). 


By the above, the number of shifted indices $j \in  J$ is at most $4n$. Hence, the longest contiguous interval without shifted indices has length at least $|J|/(4n+1) \geq m$. Let $J'$ be this interval. 
\end{proof}

We are now almost done: we have located an interval $J'\subseteq
\{0,\ldots,Fm-1\}$ of length at least $m$ where for all $i\in\{1,\ldots,n\}$
and all $j_1,j_2\in J'$ we have the same configuration in $S_{i,j_1}$ and
$S_{i,j_2}$. We now extract an assignment from this in the natural way: if
$u_{i,5j+\ell}\in D$, for some $j\in J', \ell\in\{0,\ldots,4\}$, then we set
$x_i=\ell$. We claim this satisfies $\varphi$. Consider a constraint $c_{j'}$
of $\varphi$. There must exist $j\in J'$ such that $j' = j\bmod m$, because
$|J'|\ge m$ and $J'$ is contiguous. We therefore check $H_j$, where there
exists $\sigma$ such that $Z_{\sigma,j}\subseteq I$ (this is because $j$ is not
problematic, that is, $H_j$ attains the minimum cost). But because the vertices
and incident edges of $Z_{\sigma,j}$ are dominated, it must be the case that
the assignment we extracted agrees with $\sigma$, hence $c_{j'}$ is satisfied.
\end{proof}

We now show that the pathwidth of $G$ is at most $n + O(1)$. 

\begin{lemma}\label{lem:pathwidth} The pathwidth of $G$ is at most $n+O(q5^q)$.
\end{lemma}

\begin{proof} We will show how to build a path decomposition. First, we can add $s$ to all
bags, so we focus on the rest of the graph. Second, after removing $s$ from the
graph, some vertices become leaves. It is a well-known fact that removing all
leaves from a graph can only increase the pathwidth by at most $1$. To see
this, let $G'$ be the graph obtained after deleting all leaves of $G$ and
suppose we have a path decomposition of $G'$ of width $w$. We obtain a path
decomposition of $G$ by doing the following for every leaf $v$: find a bag of
width at most $w$ that contains the neighbor of $v$ and insert after this bag,
a copy of the bag with $v$ added. Clearly, the width of the new decomposition
is at most $w+1$. Because of the above we will ignore all vertices of $G$ which
become leaves after the removal of $s$.

For all $j \in \{ 0, \ldots, Fm-1 \}$, let $S_{j}$ be defined as in Equation \ref{equation:1}.
We will show how to build a path decomposition of $G[S_j]$ with the following properties:

\begin{itemize}

\item The first bag of the decomposition contains vertices $u_{i,5j}$, for all
$i\in\{1,\ldots,n\}$.

\item The last bag of the decomposition contains vertices $u_{i,5j+4}$, for all
$i\in\{1,\ldots,n\}$.

\item The width of the decomposition is $n+O(q5^q)$.

\end{itemize}

If we achieve the above then we can obtain a path decomposition of the whole
graph: indeed, the sets $S_j$ partition all remaining vertices of the graph,
while the only edges not covered by the above decompositions are those between
$u_{i,5j+4}$ and $u_{i,5(j+1)}$. We therefore place the decompositions of $S_j$
in order, and then between the last bag of the decomposition of $S_j$ and the
first bag of the decomposition of $S_{j+1}$ we have $2n$ ``transition'' bags,
where in each transition step we add a vertex $u_{i,5(j+1)}$ in the bag, and
then remove $u_{i,5j+4}$.

Let us now show how to obtain a decomposition of $G[S_j]$, having fixed the contents of the first and last bag. First, $H_j$ has order $O(q 5^q)$, so we place all its vertices to all bags. 
The remaining graph is a union of paths of
length $4$ with the $Q_{i,j}$ gadgets attached. 
We therefore have a sequence of
$O(n)$ bags, where for each $i\in\{1,\ldots,n\}$ we add to the current bag the
vertices of $S_{i,j}$, then add and remove one after another whole copies of
$Q_{i,j}$, then remove $S_{i,j}$ except for $u_{i,5j+4}$. 
\end{proof}

We are now ready to present the main result of this section. By putting together Lemmas \ref{lemmacsptomds}, \ref{lemmamdstocsp}, \ref{lem:pathwidth} and the negative result for \textsc{$q$-CSP-5} (Lemma \ref{lemmacsp}), we get the following theorem:

\begin{theorem}\label{theoremnegativetreewidth} Under SETH, for all
$\varepsilon > 0$, no algorithm solves \textsc{Mixed Dominating Set} in time
$O^*((5-\varepsilon)^{pw})$, where $pw$ is the input graph's pathwidth.
\end{theorem}

\begin{proof}
Fix $\varepsilon > 0$ and let $q$ be sufficiently large so that \Cref{lemmacsp}
is true. Consider an instance $\varphi$ of $q$-\textsc{CSP}-5. Using our
reduction, create an instance $(G,k)$ of \textsc{Mixed Dominating Set}.  Thanks to \Cref{lemmacsptomds}
and \Cref{lemmamdstocsp}, we know that $\varphi$ is satisfiable if and only if
there exists a mixed dominating set of size at most $k$ in $G$.

Suppose there exists an algorithm which solves \textsc{Mixed Dominating Set} in time
$O^*((5-\varepsilon)^{pw})$. With this algorithm and our reduction, we can
determine if $\varphi$ is satisfiable in time $O^*((5-\varepsilon)^{pw})$,
where $pw = n+O(q5^q) = n+O(1)$, so the total running time of this procedure is
$O^*((5-\varepsilon)^n)$, contradicting the SETH.
\end{proof}



\section{Exact Algorithm}

In this section, we describe an algorithm for the
\textsc{Mixed Dominating Set} problem running in time $O^*(1.912^n)$.  Let us first give an overview
of our algorithm. Consider an instance $G = (V,E)$ of the \textsc{Mixed Dominating Set} problem and
fix, for the sake of the analysis, an optimal solution which is a nice mixed
dominating set $D \cup M$. Such an optimal solution must exist by Lemma \ref{lemmanice},
so suppose it gives the nice mds partition $V = D \cup P \cup I$. 

By \Cref{lemmavertexcover}, there exists a minimal vertex cover $C$ of $G$ for
which $D \subseteq C \subseteq D \cup P$. Our first step is to ``guess'' $C$,
by enumerating all minimal vertex covers of $G$. This decreases our search
space, since  we can now assume that vertices of $C$ only belong in $D\cup P$,
and vertices of $V\setminus C$ only belong in $P\cup I$.

For our second step, we branch on the vertices of $V$, placing them in $D$,
$P$, or $I$. The goal of this branching is to arrive at a situation where our
partial solution dominates  $V\setminus C$.  The key idea is that any vertex of
$C$ that may belong in $D$  must have at least two private neighbors, hence
this allows us to significantly speed up the branching for low-degree vertices
of $D$.  
Finally, once we have a partial solution that dominates all of $V \setminus C$, we show how to complete this optimally in polynomial time using a minimum edge cover computation.

We now describe the three steps of our algorithm in order and give the
properties we are using step by step. In the remainder we assume that $G$ has
no isolated vertices (since these are handled by taking them in the solution). Therefore, by
\Cref{lemmanice} there exists an optimal nice mixed dominating set. Denote the corresponding
partition as $V=D\cup P\cup I$.

\textbf{Step 1:} Enumerate all minimal vertex covers of $G$, which takes time $O^*(3^{n/3})$ by a result of \cite{moon1965cliques}. For each
such vertex cover $C$ we execute the rest of the algorithm. In the end output
the best solution found.

Thanks to \Cref{lemmavertexcover}, there exists a minimal vertex cover $C$ with
$D\subseteq C\subseteq D\cup P$. Since we will consider all minimal vertex
covers, in the remainder we focus on the case where the set $C$ considered
satisfies this property. Let $Z = V \setminus C$.  Then $Z$ is an independent
set of $G$.  We now get two properties we will use in the branching step of our
algorithm: 

\begin{enumerate}
\item For all $u \in C$, $u$ can be either in $D$ or in $P$, because $C
\subseteq D \cup P$. 
\item For all $v \in Z$, $v$ can be either in $P$ or in $I$, because $D
\subseteq C$. 
\end{enumerate}

\textbf{Step 2:} Branch on the vertices of $V$ as described below.

The branching step of our algorithm will be a set of Reduction and Branching
Rules over the vertices of $C$ or $Z$. In order to describe a recursive
algorithm, it will be convenient to consider a slightly more general version of
the problem: in addition to $G$, we are given three disjoint sets $D_f, P_f,
P'_f\subseteq V$, and the question is to build a nice mds partition $V = D \cup
P \cup I$ of minimum cost which satisfies the following properties: $D_f
\subseteq D\subseteq C$, $P_f\subseteq P\cap C$, and  $P'_f \subseteq P\cap Z$.
Clearly, if $D_f=P_f=P_f'=\emptyset$ we have the original problem and all
properties are satisfied. We will say that a branch where all properties are
satisfied is \emph{good}, and our proof of correctness will rely on the fact that
when we branch on a good instance, at least one of the produced branches is
good. The intuitive meaning of these sets is that when we decide in a branch
that a vertex belongs in $D$ or in  $P$ in the optimal partition we place it
respectively in $D_f$, $P_f$ or $P_f'$ (depending on whether the vertex belongs
in $C$ or $Z$).

We now describe a series of Rules which, given an instance of \textsc{Mixed Dominating Set} and three
sets $D_f, P_f, P'_f$, will recursively produce subinstances where vertices are
gradually placed into these sets. Our algorithm will consider the Reduction and
Branching Rules in order and apply the first Rule that can be applied. Note
that we say that a vertex $u$ is \emph{decided} if it is in one of the sets
$D_f \subseteq D$, $P_f \subseteq P$, or $P'_f \subseteq P$. All the other
vertices are considered \emph{undecided}. 

Throughout the description that follows, we will use $U$ to denote the set of
undecided vertices which are not dominated by $D_f$, that is, $U = V \setminus
(D_f \cup P_f \cup P'_f \cup (N(D_f)\cap Z))$. We will show that when no rule
can be applied, $U$ is empty, that is, all vertices are decided or dominated by
$D_f$.  In the third step of our algorithm we will show how to complete the
solution in polynomial time when $U$ is empty.  Since our Rules do not modify
the graph, we will describe the subinstances we branch on by specifying the
tuple $(D_f, P_f, P'_f)$.

To ease notation, let $U_C = U\cap C$ and $U_Z = U\cap Z$. Recall that for $u \in V$, we use $d_{U_C}(u)$ and $d_{U_Z}(u)$ to denote the size of the sets $N(u) \cap U_C = N_{U_C}(u)$ and $N(u) \cap U_Z = N_{U_Z}(u)$, respectively. 

\begin{figure}[h]
\begin{center}
\begin{tikzpicture}
\draw (1,6) arc (180:0:1);
\draw (1,6) -- (1,4);
\draw (3,6) -- (3,4);
\draw (1,4) arc (-180:0:1);
\draw (2,7.3) node{$C$};

\draw (4,6) arc (180:0:1);
\draw (4,6) -- (4,4);
\draw (6,6) -- (6,4);
\draw (4,4) arc (-180:0:1);
\draw (5,7.3) node{$Z$};

\draw (2,6.5) node{$D_f$};
\draw (2,3.5) node{$P_f$};
\draw (5,3.5) node{$P_f'$};

\draw (5,6.5) node{$I$};

\draw [fill = black!15]  (1,6) rectangle (3,4);
\draw (2,5) node{$U_C$};
\draw [fill = black!15]  (4,6) rectangle (6,4);
\draw (5,5) node{$U_Z$};

\draw (2.5, 6.5) -- (4.5,6.5);
\draw (2.5, 6.25) -- (4.5,6.25);
\draw (2.5,6.75) -- (4.5,6.75);





\end{tikzpicture}
\end{center}
\caption{Partition of $V = C \cup Z$, of $C = D_f \cup P_f \cup U_C$, and of $Z = P_f' \cup I \cup U_Z$ during the process of the algorithm, where $I = N(D_f) \cap Z$. The only edges drawn show that $I$ is dominated by $D_f$.}
\label{fig:CZDfpfUI}
\end{figure}
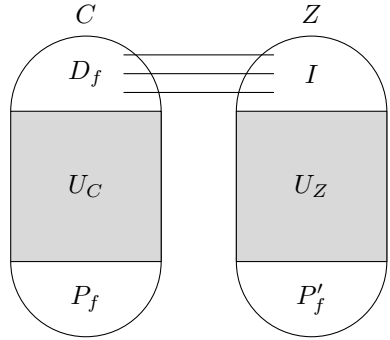

We will present each Rule individually and directly after explain why it is correct and its associated running-time, to ease presentation while having a consistent analysis.

\smallskip

\textbf{Reduction Rule (R1):} If there exists $u \in U_C$ such that
$d_{U_Z}(u) \leq 1$, then put $u$ in $P_f$, that is, recurse on the instance
$(D_f, P_f\cup\{u\}, P_f')$. 

\begin{itemize}
    \item Observe that no neighbor of $u$ in $U_C$ can be private neighbor of $u$ since $U_C \subseteq C \subseteq D \cup P$, and because $d_{U_Z}(u) \leq 1$, the vertex $u$ can have at most one private neighbor, so it must be the case that $u \in P$.
\end{itemize}

\textbf{Reduction Rule (R2):} If there exists $v \in U_Z$ such that
$d_{U_C}(v) =0$, then put $u$ in $P_f'$, that is, recurse on the instance
$(D_f, P_f, P_f'\cup\{v\})$. 

\begin{itemize}
    \item The vertex $v$ must be dominated, but it has no neighbor in $U_C$, so it must be the case that $v \in P$.
\end{itemize}

\smallskip

Now that we have presented the two Reduction Rules which we first apply in our algorithm, we will describe the Branching Rules. Thus, we need first to define our measure of progress. We define it to be the size of the set $\{ u\in U_C\ \mid
d_{U_Z}(u)\ge 2\} \cup \{ u\in U_Z\ \mid d_{U_C}(u)\ge 1\}$. In other words, we
count the undecided vertices of $U_C$ that have at least two undecided,
non-dominated vertices in $U_Z$, and the undecided, non-dominated vertices of $U_Z$
that have at least one undecided neighbor in $C$. This is motivated by the fact
that undecided vertices that do not respect these degree bounds are eliminated
by the Reduction Rules and hence do not affect the running time. Let $l$ denote
the number of the vertices that we counted according to this measure. Clearly,
$l\le n$. Let $T(l)$ be the maximum number of branches produced for an instance where the
measure has value $l$. We now consider each Branching Rule individually:

\smallskip

\textbf{Branching Rule (B1):} If there exists $u \in U_C$ such that
$d_{U_Z}(u) \geq 4$, then branch on the following two subinstances: $(D_f \cup
\{ u \}, P_f, P'_f)$ and $(D_f, P_f \cup \{ u \}, P'_f)$. 

\begin{itemize}
    \item Branching Rule B1 is correct from $U_C \subseteq C \subseteq D \cup P$. 
    \item We have $T(l) \le T(l-1) + T(l-5)$, since in the branch where $u\in D_f$ at least $4$ vertices of $U_Z$ become dominated.
    \item $T(l) \le T(l-1)+T(l-5)$ gives $x^5 = x^4+1$ with root $r<1.3248$.
\end{itemize}

\smallskip

Note that we may now assume that all vertices of $U_C$ have
$d_{U_Z}\in\{2,3\}$. The following two rules eliminate vertices $u\in U_C$ with
$d_{U_Z}(u)=2$.

\smallskip

\textbf{Branching Rule (B2.1):} If there exists $u_1,u_2 \in U_C$
such that $d_{U_Z}(u_1)=3$, $d_{U_Z}(u_2)=2$, and $N_{U_Z}(u_1)\cap
N_{U_Z}(u_2)\neq \emptyset$ then branch on the following instances:
$(D_f\cup\{u_1\},P_f\cup\{u_2\},P_f')$ and $(D_f,P_f\cup\{u_1\},P_f')$.

\begin{itemize}
    \item Branching Rule B2.1 is correct because if $u_1 \in D$, then $u_2$ cannot have two private neighbors and it is forced to be in $P$. 
    \item We have $T(l)\le T(l-1) + T(l-5)$, since in the branch where $u_1\in D_f$ we also set $u_2\in P_f$ and $3$ vertices of $U_Z$ become dominated.
    \item $T(l) \le T(l-1)+T(l-5)$ gives $x^5 = x^4+1$ with root $r<1.3248$.
\end{itemize}

\textbf{Branching Rule (B2.2):} If there exists $u \in U_C$ with
$d_{U_Z}(u)=2$ we branch on the instances $(D_f\cup\{u\},P_f,P_f')$ and
$(D_f,P_f\cup\{u\},P_f')$.

\begin{itemize}
    \item Branching Rule B2.2 is correct again from $U_C \subseteq C \subseteq D \cup P$.
    \item Let $N(u)\cap U_Z = \{v_1,v_2\}$. Note that if $d_{U_C}(v_1)\ge 2$, then all vertices $u'\in U_C$ adjacent to $v_1$ must have $d_{U_Z}(u')=2$. This is because Rules R1, B1 and B2.1 do not apply. Let $s$ be the number of vertices of $\{v_1,v_2\}$ which have at least two neighbors in $U_C$.  We consider the following cases:
	\begin{itemize}
	\item If $s=0$, then $T(l)\le T(l-3) + T(l-3)$, because when $u\in P_f$,  $v_1,v_2$ no longer contribute to $l$ (they have no other neighbor in $U_C$).
	\item If $s=1$, then $T(l)\le T(l-4) + T(l-2)$. To see this, let $u'\in U_C$ be a neighbor of $\{v_1,v_2\}$. As we said, $d_{U_Z}(u')=2$, so setting $u\in D_f$ will activate Rule R1 on $u'$, decreasing $l$ by 4. On the other hand, if $u\in P_f$, then one of $\{v_1,v_2\}$ is deleted by Rule R2.
	\item If $s=2$ and $|(N(v_1)\cup N(v_2))\cap U_C|=2$. Then we must have $N_{U_C}(v_1) = N_{U_C}(v_2) = \{u,u'\}$ for some $u'\in U_C$ with $d_{U_Z}(u')=2$. In this case Rule B2.2 (and Rules R1, R2) will be applied successively to $u,u'$ giving $T(l)\le 3T(l-4)$.
	\item If none of the above applies, then $s=2$ and we have $T(l)\le T(l-5) + T(l-1)$, because when $u\in D_f$ we force at least two other vertices of $U_C$ into $P_f$.
	\end{itemize}
    \item 
    \begin{itemize}
        \item $T(l) \le 2T(l-3)$ gives $x^3=2$ with root $r<1.2600$.
        \item $T(l) \le T(l-4)+T(l-2)$ gives $x^4=x^2+1$ with root $r<1.2721$.
        \item $T(l)\le 3T(l-4)$ gives $x^4=3$ with root $r<1.3161$.
        \item $T(l) \le T(l-5)+T(l-1)$ gives $x^5 = x^4+1$ with root $r<1.3248$.
    \end{itemize}
\end{itemize}

\smallskip

We now have that all vertices $u\in U_C$ have $d_{U_Z}(u)=3$.  Let us now
branch on vertices of $U_Z$ to ensure that these also do not have too low
degree. 

\smallskip

\textbf{Branching Rule (B3.1):} If there exists $v \in U_Z$ with
$d_{U_C}(v)=1$ let $N_{U_C}(v)=\{u\}$. We branch on the instances
$(D_f\cup\{u\},P_f,P_f')$ and $(D_f,P_f\cup\{u\},P_f')$.

\begin{itemize}
    \item Branching Rule B3.1 is correct again from $U_C \subseteq C \subseteq D \cup P$.
    \item We have $T(l)\le T(l-4) + T(l-2)$, since when $u\in P_f$ we
apply Rule R2.
    \item $T(l) \le T(l-4)+T(l-2)$ gives $x^4=x^2+1$ with root $r<1.2721$.
\end{itemize}

\textbf{Branching Rule (B3.2):} If there exists $v \in U_Z$ with
$d_{U_C}(v)=2$ let $N_{U_C}(v)=\{u_1, u_2\}$. We branch on the instances
$(D_f\cup\{u_1\},P_f,P_f')$,  $(D_f\cup\{u_2\},P_f\cup\{u_1\},P_f')$, and
$(D_f,P_f\cup\{u_1, u_2\},P_f')$.

\begin{itemize}
    \item Branching rule B3.2 is correct since we have the three following cases: $u_1 \in D$; or $u_1 \in P$ and $u_2 \in D$; or $u_1$ and $u_2 \in P$.
    \item We have $T(l)\le T(l-3) + T(l-4) + T(l-5)$ using the fact that $d_{U_Z}(u_1)=d_{U_Z}(u_2)=3$ and the fact that Rule R2 is applied when $u_1, u_2 \in P_f$.
    \item $T(l) \le T(l-3) + T(l-4) + T(l-5)$ gives $x^5=x^2+x+1$ with root $r< 1.3248$.
\end{itemize}

\smallskip

If we cannot apply any of the above Rules, for all $u\in U_C$ we have
$d_{U_Z}(u)=3$ and for all $v\in U_Z$ we have $d_{U_C}(v)\ge 3$. We now
consider three remaining cases: 
(i) two vertices of $U_C$ have two common neighbors in $U_Z$;
(ii) there exists a vertex $v\in
U_Z$ with $d_{U_C}(v)=3$; (iii) everything else.

\smallskip

\textbf{Branching Rule (B4):} If there exist $u_1,u_2 \in U_C$ and
$v_1,v_2\in U_Z$ with $(u_i,v_j)\in E$ for all $i,j\in \{1,2\}$, then we branch
on the instances $(D_f\cup\{u_1\},P_f\cup\{u_2\},P_f')$ and
$(D_f,P_f\cup\{u_1\},P_f')$.

\begin{itemize}
    \item Branching Rule B4 is correct because if $u_1 \in D$, then $u_2$ cannot have two private neighbors since $d_{U_Z}(u_2) = 3$. 
    \item We have $T(l)\le T(l-1) + T(l-5)$.
    \item $T(l) \le T(l-5)+T(l-1)$ gives $x^5 = x^4+1$ with root $r<1.3248$.
\end{itemize}

\textbf{Branching Rule (B5):} If there exists $v\in U_Z$ with
$d_{U_C}(v)=3$, let $N_{U_C}(v)=\{u_1,u_2,u_3\}$ and for $i\in\{1,2,3\}$ let
$X_i = \{ w\in U_C\setminus\{u_1,u_2,u_3\}\ |\ N(w)\cap N(u_i)\cap (U_Z \setminus \{ v \}) \neq
\emptyset\}$, that is, $X_i$ is the set of vertices of $U_C$ that share a
neighbor  with $u_i$ in $U_Z$ other than $v$.  Then we branch on the following
$8$ instances: (i) the instance $(D_f, P_f\cup\{u_1,u_2,u_3\}, P_f'\cup\{v\}\}$
(ii) for $i\in\{1,2,3\}$, we produce the instances $(D_f\cup\{u_i\},
P_f\cup(\{u_1,u_2,u_3\}\setminus\{u_i\}),P_f')$ (iii) for $i,j\in\{1,2,3\}$,
with $i<j$ we produce the instances $(D_f\cup\{u_i,u_j\},
P_f\cup(\{u_1,u_2,u_3\}\setminus\{u_i,u_j\})\cup X_i\cup X_j, P_f')$ (iv) we
produce the instance $(D_f\cup\{u_1,u_2,u_3\},P_f\cup X_1\cup X_2\cup
X_3,P_f')$.

\begin{itemize}
    \item Branching Rule B5 is correct since we have the following cases: (i) all vertices $u_1, u_2$ and $u_3$ are in $P$; (ii) or exactly one of them is in $D$; (iii) or exactly two of them are in $D$; (iv) or all of them are in $D$. Note first that $u_1$, $u_2$ and $u_3$ only share $v$ as neighbor in $U_Z$ since Branching Rule B4 is not triggered. Branching Rule B5 is correct by the following arguments:
    \begin{itemize}
        \item[(i)] $v$ must be dominated so it must be the case that $v \in P$;
        \item[(ii)] The two vertices not in $D$ necessarily are in $P$;
        \item[(iii)] Since $u_i$ and $u_j$ share $v$ as common neighbor and both have exactly three neighbors in $U_Z$, the vertices of $X_i$ and $X_j$ have to be in $P$ because otherwise $u_i$ and $u_j$ do not have two private neighbors;
        \item[(iv)] For the same reason, the vertices of $X_1$, $X_2$ and $X_3$ have to be in $P$. 
    \end{itemize}
    \item We have $T(l) \le T(l-4)+3T(l-6)+3T(l-12)+T(l-14)$. Indeed we have: (i) the branch where $u_1,u_2,u_3\in P_f$, which also effectively eliminates $v$; (ii) the branch where $u_1\in D_f$ and $u_2,u_3\in P_f$, which also dominates $N_{U_Z}(u_1)$ (plus two more symmetric branches); (iii) the branch where $u_1,u_2\in D_f$ and $u_3\in P_f$ (plus two more symmetric branches). Here we first observe that $\{v,u_1,u_2,u_3\}\cup\left( (N(u_1)\cup N(u_2))\cap U_Z\right)$ contains exactly $8$ distinct vertices, because $d_{U_Z}(u_1) = d_{U_Z}(u_2) = 3$, while $N_{U_Z}(u_1)$ and $N_{U_Z}(u_2)$ share exactly one common element ($v$), since Rule B4 does not apply. In addition to eliminating these $8$ vertices, this branch also eliminates $X_1\cup X_2$. We argue that $X_1$ alone contains at least $4$ additional vertices, distinct from the $8$ eliminated vertices. Let $N_{U_Z}(u_1) = \{v, w_1,w_2\}$. We know that $d_{U_C}(w_1), d_{U_C}(w_2)\ge 3$, since Rule B3.2 did not apply. Furthermore, since $w_1,w_2$ share $u_1$ as a common neighbor in $U_C$, they cannot share another, as Rule B4 would apply. In addition, neither $w_1$ nor $w_2$ can be connected to $u_2$ or $u_3$, since together with $v,u_1$ this would active Rule B4. Hence, we eliminate at least $12$ vertices for each of these three branches. Finally, the case (iv) where $u_1,u_2,u_3\in D_f$ is similar, except we also eliminate two additional neighbors of $u_3$ in $U_Z$ which now become dominated.
    \item $T(l) \le T(l-4) +  3T(l-6) + 3T(l-12) + T(l-14)$ gives $x^{14}=x^{10}+3x^8+3x^2+1$ with root $r< 1.3252$.
\end{itemize}

\textbf{Branching Rule (B6):} Consider $u \in U_C$ and let
$N_{U_Z}(u) = \{v_1,v_2,v_3\}$. We branch on the following instances: $(D_f,
P_f\cup\{u\}, P_f')$, $(D_f\cup\{u\}, P_f\cup (N_{U_C}(v_1)\setminus\{u\}),
P_f')$, and $(D_f\cup\{u\},P_f\cup ((N_{U_C}(v_2)\cup N_{U_C}(v_3) )\setminus
\{u\}),P_f')$.

\begin{itemize}
    \item Branching Rule B6 is correct because if $u \in D$, then either $v_1$ is one of its private neighbors, or both $v_2$ and $v_3$ are its private neighbors. 
    \item We have $T(l) \le T(l-1) + T(l-7) + T(l-10)$. Here we use the fact that since Rule B5 does not apply, $d_{U_C}(v_i) \ge 4$ and also that since Rule B4 does not apply, $N(v_i)\cap N(v_j)\cap U_C = \{u\}$ for all $i,j\in\{1,2,3\}$. Hence, the branch where $u\in D_f$ and $v_1$ is a private neighbor of $u$ forces three more vertices of $U_C$ into $P_f$, and the branch where $v_2,v_3$ are private neighbors of $u$ forces six more vertices of $U_C$ into $P_f$.
    \item $T(l) \le T(l-1) + T(l-7) + T(l-10)$ gives $x^{10} = x^9 + x^3 +1$ with
root $r< 1.3001$.
\end{itemize}

\smallskip

Our algorithm applies the above Rules in order as long as possible. Since we have proved the correctness of our Rules individually, we can explain what happens when no Rule is applicable. But first, let us establish a useful property. 

\begin{lemma}\label{lem:nou} If none of the Rules can be applied then
$U=\emptyset$. \end{lemma}

\begin{proof}
Observe that by applying rules R1, B1, B2.2, B6, we eventually eliminate all
vertices of $U_C$, since these rules alone cover all the cases for $d_{U_Z}(u)$
for any $u\in U_C$. So, if none of these rules applies, $U_C$ is empty. But
then applying R2 will also eliminate $U_Z$, which makes all of $U$ empty.
\end{proof}

\textbf{Step 3:} When $U$ is empty, reduce the problem to \textsc{Edge Cover}.

We now show how to complete the solution in polynomial time.

\begin{lemma}\label{lemmaconstructionpoly}
Let $(D_f, P_f, P_f')$ be a good tuple such that no Rule can be applied. Then it is possible to construct in polynomial time a mixed dominating set of size at most $|D| + |M|$.
\end{lemma}


\begin{proof}
Because no Rule can be applied, by Lemma \ref{lem:nou}, we have that $U = V \setminus ( D_f \cup P_f \cup P_f' \cup (N(D_f) \setminus C)) = \emptyset$.

Let $M'$ be a minimum edge cover of $G[P_f \cup P_f']$. Then, we claim that $|D| + |M| \geq |D_f| + |M'|$. First, $|D| \geq |D_f|$ because $D_f \subseteq D$. We now claim that $|M| \geq |M'|$. Note that $P_f \subseteq P \cap C$ and $P_f' \subseteq P \cap Z$, so $P_f \cup P_f' \subseteq P$. $M$ is an edge cover of $G[P]$, and $M'$ is a minimum edge cover of $G[P_f \cup P_f']$, with $P_f \cup P_f' \subseteq P$, so necessarily $|M| \geq |M'|$. 

Since \textsc{Edge Cover} is in P, we have constructed in polynomial time a mixed dominating set of size at most $|D| + |M|$. 
\end{proof}


We can now prove the main result of this section :

\begin{theorem}\label{theoremexact} \textsc{Mixed Dominating Set} can be solved in time
$O^*(1.912^n)$ and polynomial space.  \end{theorem}

\begin{proof}
The algorithm first enumerates all minimal vertex covers $C$; then applies all
Rules exhaustively; and then for each branch invokes
\Cref{lemmaconstructionpoly}. In the end we output the best solution found. 

By \Cref{lemmanice} and \Cref{lemmavertexcover} we obtain (assuming we have already taken isolated vertices) that there exists an optimal nice
mds partition $V=D\cup P\cup I$ and a minimal vertex cover with $D\subseteq
C\subseteq D\cup P$, so consider the execution of the algorithm on $C$. 
We have proven that one of the branches will end up with a good tuple, 
and by
\Cref{lem:nou} when we can no longer apply any Rules, $U$ is empty, so we
correctly solve the resulting instance in polynomial time by
\Cref{lemmaconstructionpoly}. Hence, the algorithm produces a correct solution.

Let us now analyze the running time. First, enumerating all minimal vertex cover takes time at most $O^*(3^{n/3})$, which is also an upper bound on the number of such covers by a result of \cite{moon1965cliques}. Moreover, we observe that we can decide if a Rule applies in polynomial time, and the algorithm of \Cref{lemmaconstructionpoly} runs in polynomial time. We therefore only need to bound the number of subinstances the branching step will produce, as a function of $n$. 

Of all the branching vectors, the worst case is given by Branching Rule B5, which leads to a complexity of $1.3252^l$. Taking into account the cost of enumerating all minimal vertex covers and the fact that $l \leq n$, the running time of our algorithm is $O^*(3^{n/3} \cdot 1.3252^n) = O^*(1.912^n)$. 
\end{proof}
\section{FPT Algorithm}

In this section, we describe an algorithm for the \textsc{Mixed Dominating
Set} problem parameterized by the solution size $k$ running in time $O^*(3.510^k)$. Let us give an overview of our
algorithm. Consider an instance $(G = (V,E), k)$ of the \textsc{Mixed
Dominating Set} problem parameterized by $k$, and fix, for the sake of the analysis, a solution of
size $k$ which is a nice mixed dominating set. If a solution of size $k$
exists, then a nice solution must exist by Lemma \ref{lemmanice} (assuming
without loss of generality that $G$ has no isolated vertices), so suppose it
gives the nice mds partition $V = D \cup P \cup I$. Note that for such a solution $D \cup M$ of size $k$, we have $k \geq |D| + |P|/2$ since $|M| \geq |P|/2$. 

Our algorithm begins by performing a branching step similar in spirit to that
of the algorithm of \Cref{theoremexact}, trying to guess a part of this
partition.  In particular, we gradually build up two disjoint sets $D_f, P_f$
which store the vertices that must belong to $D$ and $P$ respectively. Let
$U=V\setminus (D_f\cup P_f)$ be the set of ``undecided vertices'' and,
furthermore, let $U^*=U\setminus N(D_f)$ be the set of undecided vertices which
are not currently dominated by the solution. Our algorithm proceeds in the
following steps: (i) first, we branch with the goal of eliminating $U^*$, that
is, with the goal of finding a partial solution that dominates all vertices
(ii) then, because the considered solution is nice, we observe that we cannot
place any more vertices in $D_f$; we therefore perform a simple ``vertex
cover''-type branching in $G[U]$, until we arrive at a situation where the
maximum degree of $G[U]$ is $1$ (iii) then, we invoke a result of
\cite{XiaoS19} to complete the solution in polynomial time.  As with the
algorithm of \Cref{theoremexact}, we use the fact that the sought solution is
nice to speed up the branching on low-degree vertices.

\noindent \textbf{Step 1:} Branch to eliminate $U^*$.

Recall that we have fixed for the analysis an optimal nice mds partition
$V=D\cup P\cup I$.  As with the algorithm of \Cref{theoremexact}, it will be
convenient to describe a recursive algorithm which is given two disjoint sets
of vertices $D_f,P_f$. We will say that the sets $(D_f,P_f)$ are \emph{good} if
$D_f\subseteq D$ and $P_f\subseteq P$. Clearly, these conditions are satisfied
if $D_f=P_f=\emptyset$. We will describe a series of Rules, which must be
applied exhaustively, always selecting the first Rule that can be applied.  For
correctness, we will show that for each branching Rule, if the current instance
is characterized by a good pair $(D_f,P_f)$, at least one of the produced
instances is also good. When no Rule can be applied, we will proceed to the
next step.

Recall that we denote $U=V\setminus (D_f\cup P_f)$ and $U^*=U\setminus N(D_f)$.
Our strategy will be to branch in a way that eliminates $U^*$ as quickly as
possible because, as we will see in the next step, once this is done the
problem becomes much easier. We begin branching from low-degree vertices, which
will allow us to assume that all remaining vertices are high-degree as we
consider later Rules. Here, for a vertex $u\in U^*$ we are mostly interested in
its degrees in $G[U]$ and in $G[U^*]$. Note that $d_{U^*}(u) \le
d_{U}(u)$ since $U^* \subseteq U$.

As in the algorithm of \Cref{theoremexact}, we will present each Rule individually and directly after explain why it is correct and its associated running-time. 
Recall for the sake of the analysis that an instance is good if $D_f \subseteq D$ and $P_f \subseteq P$.

\smallskip

\noindent \textbf{Sanity Check Rule:} If $|D_f|+\frac{|P_f|}{2}>k$, reject. If
there exists $u\in D_f$ such that 
$|N_{U^*}(u)| \leq 1$, reject.

\begin{itemize}
    \item The Sanity Check Rule will reject if either the currently decided vertices in $D_f$ and $P_f$ have total cost more than $k$ (which implies that this is not a good instance, as the correct partition has cost at most $k$); or if a vertex $u\in D_f$ has at most one private neighbor in $U^*$. Since the number of private neighbors of $u$ can only diminish if we add vertices to $D_f$, if $u\in D$ this would contradict the niceness of the partition $D\cup P\cup I$. Hence, in this case also the current instance is not good.
\end{itemize}

\noindent \textbf{Reduction Rule (R1):} If there exists $u \in U^*$ with
$d_U(u) = 0$, then put $u$ in $P_f$, that is, recurse on the instance $(D_f,
P_f\cup\{u\})$. 

\begin{itemize}
    \item If the current instance is good, then $u\not\in D$ (because it would not have two private neighbors) and $u\not\in I$ (because it would not be dominated). Hence, the new instance is also good.
\end{itemize}

\smallskip


Now that we have presented our Sanity Check Rule and our only Reduction Rule which we first apply in our algorithm, we will describe the Branching Rules. Thus, we need first to define our measure of progress. We define it to be $l = 2k - 2|D_f| - |P_f|$. Initially, $l = 2k$, and we observe that because of the Sanity Check Rule in all produced instances we have $l \geq 0$. We will therefore upper bound the number of produced instances by measuring how much each Branching Rule decreases $l$. Let $T(l)$ be the maximum number of branches produced for an instance where measure has value $l$. We now consider each Branching Rule individually:


\smallskip

\noindent \textbf{Branching Rule (B1):} If there exists $u \in U^*$ with
$d_U(u) = 1$, then let $N_U(u) = \{ v \}$. Branch on the following two
subinstances: $(D_f, P_f \cup \{ u \})$ and $(D_f \cup \{ v \}, P_f)$. 

\begin{itemize}
    \item We note that $u\not\in D$, because it would not have two private neighbors. If $u\in I$, then $v\in D$, because $u$ must be dominated. Hence, one of the branches is good.
    \item We have $T(l)\le T(l-1)+T(l-2)$.
    \item $T(l) \leq T(l-1) + T(l-2)$ gives $x^2 = x + 1$ with root $r < 1.6181$.
\end{itemize}

\smallskip

We are now at a situation where all vertices $u \in U^*$ have $d_U(u) \geq 2$. 

\smallskip

\noindent \textbf{Branching Rule (B2.1):} If there exists $u \in U^*$ with
$d_U(u) = 2$ and $d_{U^*}(u) \in \{ 0, 1 \}$, then let $N_U(u) = \{ v_1, v_2 \}$.
Branch on the subinstances: $(D_f, P_f \cup \{ u \})$, $(D_f \cup \{ v_1 \},
P_f \cup \{ v_2 \})$, $(D_f \cup \{ v_2 \}, P_f \cup \{ v_1 \})$ and $(D_f \cup
\{ v_1, v_2 \}, P_f)$.  

\begin{itemize}
    \item We again have $u\not\in D$, because it would not have two private neighbors. If $u\in I$, then $\{v_1,v_2\}\subseteq D\cup P$ and we consider all such possibilities, except $v_1,v_2\in P$, because $u$ must be dominated.
    \item We have $T(l)\le T(l-1) + 2T(l-3) + T(l-4)$.
    \item $T(l) \leq T(l-1) + 2T(l-3) + T(l-4)$ gives $x^4 = x^3 + 2x + 1$ with root $r < 1.7944$.
\end{itemize}

\smallskip

Before presenting Branching Rule 2.2, we make a simple observation we will use several times in the rest of the algorithm: for a vertex $u \in U^*$ with two neighbors $v_1, v_2 \in U$, if we put $u$ in $D_f$ with $v_1$ and $v_2$ its private neighbors, then we must put $(N_U(v_1) \cup N_U(v_2)) \setminus \{ u \}$ in $P_f$ in order to have $v_1$ and $v_2$ the private neighbors of $u$. 

Let us also introduce another helpful definition. For $v_1,v_2\in U$, we will say
that $\{v_1,v_2\}$ is a \textit{feasible pair} if $|N_{U^*}(v_1)\setminus N(v_2)|\ge
2$ and $|N_{U^*}(v_2)\setminus N(v_1)|\ge 2$. Informally, $\{v_1,v_2\}$ is
feasible if it is possible that both $v_1,v_2\in D$. In other words,
$\{v_1,v_2\}$ is \emph{not} feasible, if placing both vertices in $D_f$ would
immediately activate the Sanity Check rule because one of the two vertices
would not have enough private neighbors.

\smallskip

\noindent \textbf{Branching Rule (B2.2):} If there exists $u \in U^*$ with
$d_U(u) = 2$, then let $N_U(u) = \{ v_1, v_2 \}$. Branch on the following
subinstances: $(D_f \cup \{ u \}, P_f \cup ((N_U(v_1)  \cup N_U(v_2)) \setminus
\{ u \}))$, $(D_f, P_f \cup \{ u \})$, $(D_f \cup \{ v_1 \}, P_f \cup \{ v_2
\})$, $(D_f \cup \{ v_2 \}, P_f \cup \{ v_1 \})$ and $(D_f \cup \{ v_1, v_2 \},
P_f)$.

\begin{itemize}
    \item We have the same cases as before, but now it is possible that $u\in D$. However, in this case $v_1,v_2$ must be private neighbors of $u$, hence $(N_U(v_1)\cup N_U(v_2))\setminus \{u\}$ must be a subset of $P$.
    \item Since B2.1 does not apply, we have $v_1,v_2\in U^*$, so $d_U(v_1),d_U(v_2)\ge 2$.  We consider the following cases:
	\begin{itemize}
	    \item If $\{v_1,v_2\}$ is not a feasible pair, that is (without loss of generality) we have $|N_{U^*}(v_1)\setminus N(v_2)|\le 1$. First, we note that $d_U(v_1)\ge 2$, since Rule B1 did not apply, so $|N_U(v_1)\setminus \{u\}|\ge 1$. Also, the Sanity Check Rule is activated for the instance where $v_1,v_2\in D_f$.  Taking into account the remaining instances we have: $T(l)\le T(l-1) + 3T(l-3)$.
	    \item If $\{v_1,v_2\}$ is a feasible pair then $|(N_{U^*}(v_1)\cup N_{U^*}(v_2))\setminus \{u\}|\ge 4$, because each of $v_1,v_2$ has two non-dominated neighbors which are not neighbors of the other. We have $T(l) \le T(l-1)+2T(l-3)+T(l-4)+T(l-6)$.
	\end{itemize}
	\item 
	\begin{itemize}
	    \item $T(l) \leq T(l-1) + 3T(l-3)$ gives $x^3 = x^2 + 3$ with root $r < 1.8638$.
	    \item $T(l) \leq T(l-1) + 2T(l-3) + T(l-4) + T(l-6)$ gives $x^6 = x^5 + 2x^3 + x^2 + 1$ with root $r < 1.8199$.
	\end{itemize}
\end{itemize}

\smallskip

We are now at a situation where all vertices $u \in U^*$ have $d_U(u) \geq 3$.

\smallskip

\noindent \textbf{Branching Rule (B3.1):} If there exists $u \in U^*$ with
$d_U(u) = 3$ and $d_{U^*}(u) \le 2$, then let $N_U(u) = \{v_1,v_2,v_3\}$ and
let $v_3\in U\setminus U^*$. We branch on $(D_f, P_f\cup\{u\})$,
$(D_f\cup\{u\}, P_f\cup ((N_U(v_1)\cup N_U(v_2))\setminus \{u\}))$, and for each
non-empty $S\subseteq \{v_1,v_2,v_3\}$ we branch on $(D_f\cup S,
P_f\cup(\{v_1,v_2,v_3\}\setminus S))$.

\begin{itemize}
    \item We observe that if $u\in D$, then $v_1,v_2$ must be its private neighbors, so again $(N_U(v_1)\cup N_U(v_2))\setminus \{u\}$ must be a subset of $P$. If $u\in I$ we consider all partitions of $N_U(u)$ into $D$ and $P$, while ensuring that $u$ is dominated.
    \item First note that if $v_i\in U^*$, then $d_U(v_i)\ge 3$, since Rules B1-B2.2 do not apply, hence $|N_U(v_i)\setminus\{u\}|\ge 2$. Consider the following subcases:
	\begin{itemize}
	    \item $d_{U^*}(u)\le 1$. Then the branch where $u\in D_f$ is immediately eliminated by the Sanity Check Rule. For the remaining branches we have $T(l) \le T(l-1) + 3T(l-4) + 3T(l-5) + T(l-6)$.
	    \item $d_{U^*}(u)=2$, so $N_{U^*}(u) = \{ v_1, v_2 \}$, and $\{v_1,v_2\}$ is a feasible pair. Then $|(N_U(v_1)\cup N_U(v_2))\setminus\{u\}| \ge 4$. We have $T(l)\le T(l-1)+ 3T(l-4) + 3T(l-5) + 2T(l-6)$
	    \item $d_{U^*}(u)=2$, so $N_{U^*}(u) = \{ v_1, v_2 \}$, and $\{v_1,v_2\}$ is not a feasible pair. Then the Sanity Check Rule eliminates the instances where $\{ v_1, v_2 \} \subseteq D$. Taking into account the remaining instances we have $T(l)\le T(l-1) + 4T(l-4) + 2T(l-5)$.
	\end{itemize}
	\item 
	\begin{itemize}
	    \item $T(l) \leq T(l-1) + 3T(l-4) + 3T(l-5) + T(l-6)$ gives $x^6 = x^5 + 3x^2 + 3x +1 $ with root $r < 1.8205$.
        \item $T(l) \leq T(l-1) + 3T(l-4) + 3T(l-5) + 2T(l-6)$ gives $x^6 = x^5 + 3x^2 + 3x + 2$ with root $r < 1.8393$.
        \item $T(l) \leq T(l-1) + 4T(l-4) + 2T(l-5)$ gives $x^5 = x^4 + 4x + 2$ with root $r < 1.8305$.
	\end{itemize}
\end{itemize}

\noindent \textbf{Branching Rule (B3.2):} If there exists $u\in U^*$ with
$d_U(u)=3$ such that there exist at least two feasible pairs in $N_U(u)$, then
we do the following. Let $N_U(u) = \{v_1,v_2,v_3\}$. For each
$i,j\in\{1,2,3\}$, with $i<j$, branch on $(D_f\cup\{u\}, P_f\cup ((N_U(v_i)\cup
N_U(v_j))\setminus \{u\}))$. Furthermore, branch on the instances $(D_f,
P_f\cup\{u\})$ and, for each non-empty $S\subseteq \{v_1,v_2,v_3\}$ on the
instance $(D_f \cup S, P_f \cup( \{v_1,v_2,v_3\} \setminus S))$.

\begin{itemize}
    \item We branch in a similar fashion as in Branching 3.1, except that for all $i,j\in\{1,2,3\}$ with $i<j$, we consider the case that $v_i,v_j$ are private neighbors of $u$, when $u\in D$.
    \item We have $v_1,v_2,v_3\in U^*$. We again note that if $\{v_i,v_j\}$ is feasible then $|(N_U(v_i)\cup N_U(v_j))\setminus \{u\}|\ge 4$. Therefore, a branch corresponding to a feasible pair diminishes $l$ by at least $6$. We consider the subcases:
	\begin{itemize}
	    \item All three pairs from $\{v_1,v_2,v_3\}$ are feasible. Then we get $T(l)\le T(l-1) + 3T(l-4)+ 3T(l-5) + 4T(l-6)$.
	    \item Two of the pairs from $\{v_1,v_2,v_3\}$ are feasible, and one, say $\{v_1,v_2\}$ is not feasible. Then the Sanity Check rule will eliminate the branches that have $\{ v_1, v_2 \} \subseteq D$. We therefore get $T(l)\le T(l-1) + 4T(l-4) + 2T(l-5) + 2T(l-6)$.
	\end{itemize}
	\item
	\begin{itemize}
	    \item $T(l) \leq T(l-1) + 3T(l-4) + 3T(l-5) + 4T(l-6)$ gives $x^6 = x^5 + 3x^2 + 3x + 4$ with root $r < 1.8734$. 
        \item $T(l) \leq T(l-1) + 4T(l-4) + 2T(l-5) + 2T(l-6)$ gives $x^6 = x^5 + 4x^2 + 2x + 2$ with root $r < 1.8672$. 
	\end{itemize}
\end{itemize}

\noindent \textbf{Branching Rule (B3.3):} If there exists $u\in U^*$ with
$d_U(u)=3$, let $N_U(u)=\{v_1,v_2,v_3\}$ and branch on the following instances:
$(D_f, P_f\cup\{u\})$, $(D_f\cup\{u\}, P_f\cup (N_U(v_1)\setminus \{u\}))$,
$(D_f\cup\{u\}, P_f\cup ((N_U(v_2)\cup N_U(v_3))\setminus \{u\}))$, and for each
non-empty $S\subseteq \{v_1,v_2,v_3\}$ branch on $(D_f\cup S,
P_f\cup(\{v_1,v_2,v_3\}\setminus S))$.

\begin{itemize}
    \item We make a variation of the previous branching by arguing that if $u\in D$ then either $v_1$ is its private neighbor, or both $v_2,v_3$ are its private neighbors.
    \item Now at least two pairs in $\{v_1,v_2,v_3\}$ are not feasible (otherwise we would have applied B3.2). Then, the Sanity Check Rule eliminates all branches where $D$ contains an infeasible pair. From the remaining branches we get $T(l)\le T(l-1) + 5T(l-4) + T(l-5)$.
    \item $T(l) \leq T(l-1) + 5T(l-4) + T(l-5)$ gives $x^5 = x^4 + 5x + 1$ with root $r < 1.8603$.
\end{itemize}

\smallskip

We are now at a situation where all  vertices $u \in U^*$ have $d_U(u) \geq 4$.
The next case we would like to handle is that of a vertex $u\in U^*$ with
$d_U(u)=d_{U^*}(u)=4$. For such a vertex let $N_U(u) =\{v_1,v_2,v_3,v_4\}$.
Let us now give one more helpful definition. For some $i\in\{1,2,3,4\}$, we
will say that $v_i$ is \emph{compatible for} $u$, if $N_{U^*}(v_i)$ contains at
least two vertices which are not neighbors of any $v_j$, for $j\in
\{1,2,3,4\}\setminus\{i\}$. In other words, $v_i$ is compatible if it has two
private neighbors, which will remain private even if we put all of
$\{v_1,v_2,v_3,v_4\}$ in $D_f$. Using this definition, we distinguish the
following two cases:

\smallskip

\noindent \textbf{Branching Rule (B4.1):} If there exists $u \in U^*$ with
$d_U(u) = d_{U^*}(u) =  4$, $N_U(u)=\{v_1,v_2,v_3,v_4\}$, and all $v_i$ are
compatible for $u$, then we branch on the following instances:
$(D_f,P_f\cup\{u\})$; for each $i,j\in\{1,2,3,4\}$ with $i<j$ we branch on
$(D_f\cup\{u\}, P_f\cup((N_U(v_i)\cup N_U(v_j))\setminus \{u\}))$; for each
non-empty subset $S \subseteq N_U(u)$, let $S^c = N_U(u) \setminus S$, branch
on $(D_f \cup S, P_f \cup S^c)$. 

\begin{itemize}
    \item The branching is similar to B3.2: if $u\in D$, then two of its neighbors must be private and we consider all possibilities.
    \item If all $v_i\in U^*$ and all $v_i$ are compatible for $u$ that means that for all $i,j\in\{1,2,3,4\}$, with $i<j$, we have $|(N_U(v_i)\cup N_U(v_j))\setminus \{u\}|\ge 5$, where we use that $v_i$ has at least two neighbors in $U$ which are not connected to $v_j$ (and vice-versa) and that $d_U(v_i), d_U(v_j)\ge 4$ since previous Rules do not apply.  We therefore have $T(l)\le T(l-1) + 4T(l-5) + 6T(l-6) + 10T(l-7) + T(l-8)$.
    \item $T(l) \leq T(l-1) + 4T(l-5) + 6T(l-6) + 10T(l-7) + T(l-8)$ gives $x^8 = x^7 + 4x^3 + 6x^2 + 10x + 1$ with root $r < 1.8595$. 
\end{itemize}

\smallskip

If the previous rule does not apply, vertices $u\in U^*$ with $d_U(u)=4$ have
either $d_{U^*}(u)\le 3$ or a neighbor $v_i\in N_U(u)$ is not compatible
for $u$.

\smallskip

\noindent \textbf{Branching Rule (B4.2):} If there exists $u \in U^*$ with
$d_U(u) = 4$, then let $N_U(u)=\{v_1,v_2,v_3,v_4\}$. Suppose without loss of
generality that for all  $1\le j \le d_{U^*}(u)$ we have $v_j\in U^*$ (that is,
vertices of $N_{U^*}(u)$ are ordered first), and that if there exists a
feasible pair in $N_{U^*}(u)$, then $\{v_1,v_2\}$ is feasible. We produce the
instances: $(D_f, P_f \cup \{ u \})$; for each non-empty subset $S \subseteq
N_U(u)$, let $S^c = N_U(u) \setminus S$, we branch on $(D_f \cup S, P_f \cup
S^c)$; if $d_{U^*}(u)\ge 2$ we produce the branch $(D_f\cup\{u\}, P_f\cup
((N_U(v_1)\cup N_U(v_2))\setminus \{u\}))$; for $3\le j \le d_{U^*}(u)$ we produce
the branch $(D_f\cup\{u\}, P_f\cup (N_{U}(v_j)\setminus\{u\}))$.

\begin{itemize}
    \item Either $u\in P$ (which we consider), or $u\in I$, so we consider all partitions of $N_U(u)$ into $D,P$ that dominate $u$, or $u\in D$. For the latter to happen it must be the case that $d_{U^*}(u)\ge 2$. In that case, either $v_1,v_2$ are both private neighbors, or $v_3$ is a private neighbor (if $v_3\in U^*$), or $v_4$ is a private neighbor (if $v_4\in U^*$).
    \item 
    \begin{itemize}
        \item We first handle the case where $d_{U^*}(u)\le 3$. Recall that, if $v_i\in U^*$, then $d_U(v_i)\ge 4$ (otherwise one of the previous Rules applies), so the (at most two) branches where $u\in D_f$ diminish $l$ by at least $5$.  We have $T(l)\le T(l-1) + 6T(l-5) + 6T(l-6) + 4T(l-7) + T(l-8)$. 
        \item If $d_{U^*}(u)=4$, we note that it cannot be the case that $\{v_1,v_2,v_3,v_4\}\subseteq D$, since this would mean that all $v_i$ are compatible for $u$ and Rule B4.1 would have applied. The branch corresponding to $S=\{v_1,v_2,v_3,v_4\}$ is therefore eliminated by the Sanity Check Rule. We consider two subcases:
	    \begin{itemize}
	        \item At least one feasible pair exists in $N_{U^*}(u)$, therefore, $\{v_1,v_2\}$ is a feasible pair.  Then $T(l)\le T(l-1) + 6T(l-5) + 6T(l-6) + 5T(l-7)$. 
	        \item No feasible pair exists. In this case the Sanity Check Rule eliminates all sets $S\subseteq N_U(u)$ that contain two or more vertices.  We have $T(l)\le T(l-1) + 7T(l-5)$.
	    \end{itemize} 
    \end{itemize}
    \item
    \begin{itemize}
        \item $T(l) \leq T(l-1) + 6T(l-5) + 6T(l-6) + 4T(l-7) +T(l-8)$ gives $x^8 = x^7 + 6x^3 + 6x^2 + 4x + 1$ with root $r < 1.8665$. 
        \item $T(l) \leq T(l-1) + 6T(l-5) + 6T(l-6) + 5T(l-7)$ gives $x^7 = x^6 + 6x^2 + 6x + 5$ with root $r < 1.8700$. 
        \item $T(l) \leq T(l-1) + 7T(l-5)$ gives $x^5 = x^4 + 7$ with root $r < 1.7487$. 
    \end{itemize}
\end{itemize}

\smallskip

We are now at a situation where all vertices $u \in U^*$ have $d_U(u) \geq 5$. 

\smallskip

\noindent \textbf{Branching Rule (B5):} If there exists $u \in U^*$ with
$d_U(u) \in \{ 5, 6, 7, 8 \}$, then select such a $u$ with minimum $d_U(u)$ and
let $i = d_U(u)$ and $N_U(u) = \{ v_1, \ldots, v_i \}$. Again, without loss of
generality we order the vertices of $N_{U^*}(u)$ first, that is, for $1\le j\le
d_{U^*}(u)$, we have $v_j\in U^*$.  Branch on the following: $(D_f, P_f \cup \{
u \})$; for each non-empty subset $S \subseteq N_U(u)$, let $S^c = N_U(u)
\setminus S$, branch on $(D_f \cup S, P_f \cup S^c)$; if $d_{U^*}(u)\ge 2$
branch on $(D_f \cup \{ u \}, P_f \cup ((N_U(v_{1}) \cup N_U(v_2)) \setminus \{
u \}))$; for $3\le j \le d_{U^*}(u)$, branch on $(D_f \cup \{ u \}, P_f \cup
(N_U(v_j) \setminus \{ u \}))$.

\begin{itemize}
    \item We generalize the previous branching to higher degrees in the obvious way: if $u\in D$, either the two first of its $d_{U^*}(u)$ neighbors in $U^*$ are its private neighbors, or one of its remaining $d_{U^*}(u)-2$ neighbors in $U^*$ is private.
    \item Let $i = d_U(u)$. Note that we then assume that $d_U(v_j)\ge i$ for all $j\in\{1,\ldots,d_{U^*}(u)\}$, since we selected $u$ with the minimum $d_U(u)$.  Hence the branches where $u\in D_f$ diminish $l$ by at least $i+1$. We then have $T(l) \le T(l-1) + (i-1)T(l-i-1) + \sum_{j=1}^{i} \binom{i}{j} T(l-i-j)$, which corresponds to the case where $d_{U^*}(u)=d_U(u)$.
    \item $T(l) \leq T(l-1) + (i-1)T(l-(i+1)) + \sum_{j = 1}^{i} \binom{i}{j} \cdot T(l-(i+j))$ gives $x^{2i} = x^{2i-1} + (i-1)x^{i-1} + \sum_{j = 1}^{i} \binom{i}{j} \cdot x^{i-j}$. 
    \begin{itemize}
        \item If $i = 5$: $x^{10} = x^9 + 9x^4 + 10x^3 + 10x^2 + 5x + 1$ with root $r < 1.8473$.
        \item If $i = 6$: $x^{12} = x^{11} + 11x^5 + 15x^4 + 20x^3 + 15x^2 + 6x + 1$ with root $r < 1.8104$.
        \item If $i = 7$: $x^{14} = x^{13} + 13x^6 + 21x^5 + 35x^4 + 35x^3 + 21x^2 + 7x + 1$ with root $r < 1.7816$.
        \item If $i = 8$: $x^{16} = x^{15} + 15x^7 + 28x^6 + 56x^5 + 70x^4 + 56x^3 + 28x^2 + 8x + 1$ with root $r < 1.7593$. 
    \end{itemize}
\end{itemize}

\smallskip

We are now at a situation where all vertices $u \in U^*$ have $d_U(u) \geq 9$. 

\smallskip

\noindent \textbf{Branching Rule (B6):} If there exists $u \in U^*$ with $d_U(u) \geq 9$, then let $T = \{ v_1, \ldots, v_9 \} \subseteq N_U(u)$. Branch on the feasible subinstances among the following: $(D_f \cup \{ u \}, P_f)$; $(D_f, P_f \cup \{ u \})$; for each (possibly empty) subset $S \subseteq T$, let $S^c = T \setminus S$, branch on $(D_f \cup S, P_f \cup S^c)$. 

\begin{itemize}
    \item We consider all possibilities (including the $T\subseteq P$), so one produced instance must be good.
    \item $T(l) \leq T(l-1) + T(l-2) + \sum_{j = 0}^9 \binom{9}{j} \cdot T(l-(9+j))$.  
    \item $T(l) \leq T(l-1) + T(l-2) + \sum_{j = 0}^9 \binom{9}{j} \cdot T(l-(9+j))$ gives $x^{18} = x^{17} + x^{16} + x^9 + 9x^8 + 36x^7 + 84x^6 + 126x^5 + 126x^4 + 84x^3 + 36x^2 + 9x + 1$ with root $r < 1.8640$. 
\end{itemize}

\smallskip

\noindent \textbf{Step 2:} Branch to eliminate $U$.

If none of the Rules above apply, we enter the second branching step of our
algorithm, which only involves one Rule that will be applied exhaustively. We observe that we now have $U^*=\emptyset$:

\begin{lemma}\label{lem:correct2}
If none of the Rules from R1 to B6 can be applied, then $U^* = \emptyset$.
\end{lemma}

\begin{proof}
Suppose that none of the Rules up to B6 applies, then $U^*=\emptyset$:
indeed all vertices $u\in U^*$ are handled according to whether $d_U(u)$ is $0$
(R1), $1$ (B1), $2$ (B2.2), $3$ (B3.3), $4$ (B4.2), $5,6,7$, or $8$ (B5), or
higher (B6).
\end{proof}


The fact that $U^* = \emptyset$ means
that all remaining undecided vertices belong in $P\cup I$, because they cannot
have two private neighbors. We use this observation to branch until we
eliminate all vertices of $G[U]$ with degree at least $2$.

\smallskip

\noindent \textbf{Branching Rule (B7):} If there exists $u \in U$ with $d_U(u)
\geq 2$, then branch on the following two subinstances: $(D_f, P_f \cup \{ u
\})$ and $(D_f, P_f \cup N_U(u))$. 

\begin{itemize}
    \item If $U^*=\emptyset$ then $U\subseteq P\cup I$, because no vertex of $U$ can have two private neighbors. Hence, if $u\in I$ then $N_U(u)\subseteq P$ and the Rule B7 is correct. 
    \item $T(l) \leq T(l-1) + T(l-2)$, because $|N_U(u)| \geq 2$. 
    \item $T(l) \leq T(l-1) + T(l-2)$ gives $x^2 = x + 1$ with root $r < 1.6181$.
\end{itemize}

\smallskip

\noindent \textbf{Step 3:} Complete the solution.

We are now in a situation where $U^*=\emptyset$ and $G[U]$ has maximum degree
$1$. We recall that Theorem 2 of \cite{XiaoS19} showed that this problem can
now be solved optimally in polynomial time by collapsing all edges of $G[U]$
and then performing a minimum edge cover computation. We recall the relevant
result, translated to our terminology:

\begin{lemma}[Theorem 2 of \cite{XiaoS19}]\label{lem:xiao} For an instance
characterized by $(D_f,P_f)$ which has $U^*=\emptyset$ and $G[U]$ has maximum
degree $1$, we can compute in polynomial time a minimum mixed dominating set
$D\cup M$ satisfying $D_f\subseteq D$, $P_f\subseteq V(M)$. \end{lemma}

We are now ready to put everything together to obtain the promised algorithm. 

\begin{theorem}\label{theoremfpt-app} \textsc{Mixed Dominating Set} parameterized by the size of the solution $k$ can be
solved in time $O^*(3.510^k)$.  \end{theorem}

\begin{proof} The algorithm applies the Rules exhaustively and when no Rule applies invokes
\Cref{lem:xiao}. Fix an optimal nice mds partition $V=D\cup P\cup I$. We assume
that no isolated vertices exist, so such a nice partition exists by
\Cref{lemmanice}. 

For correctness, we need to argue that if the cost $|D|+\frac{|P|}{2}$ is at
most $k$, the algorithm will indeed output a solution of cost at most $k$.
Observe that the initial instance is good, and we always produce a correct instance when we branch since all our Rules are correct, 
and by \Cref{lem:xiao} the solution
is optimally completed when no Rule applies, 
so if the optimal partition has
cost at most $k$, the algorithm will produce a valid solution of cost at most
$k$.

Let us now analyse the running time. We observe first that we can decide if a Rule applies in polynomial time, and the algorithm of \Cref{lem:xiao} runs in polynomial time. We therefore only need to bound the number of subinstances the branching step will produce, as a function of $k$.

Of all the branching vectors, the worst case if given by Branching Rule B3.2, which leads to a complexity of $1.8734^l$. Taking into account that $l \leq 2k$, the running time of our algorithm is $O^*(1.8734^l) = O^*(3.510^k)$. 
\end{proof}

\section{Conclusion}

In this paper, we study the \textsc{Mixed Dominating Set} problem from the exact and parameterized viewpoint. We prove first that the problem can be solved in time $O^*(5^{tw})$, and we prove that this algorithm and the one for pathwidth running in time $O^*(5^{pw})$ (\cite{JainJPS17}) are optimal, up to polynomial factors, under the SETH. Furthermore, we improve the best exact algorithm and the best FPT algorithm parameterized by the solution size $k$, from $O^*(2^n)$ and exponential space to $O^*(1.912^n)$ and polynomial space, and from $O^*(4.172^k)$ to $O^*(3.510^k)$, respectively. 

Concerning FPT algorithms for some parameters, an interesting study would be to obtain such tight results for other parameters such as bandwidth, now that the questions for pathwidth, treewidth and clique-width are closed. Moreover, it seems hard to improve our exact and FPT algorithms parameterized by $k$ using the standard technique of running time analysis we have used, but using the measure-and-conquer method could give improved bounds.

\nocite{*}
\bibliographystyle{abbrvnat}
\bibliography{sample-dmtcs-episciences}
\label{sec:biblio}

\end{document}